\documentclass[letterpaper,12pt]{article}
\usepackage{src/pascal_paper}
\usepackage[section]{placeins}
\usepackage{afterpage}
\usepackage{booktabs}

\usepackage{tikz}
\usetikzlibrary{shapes,snakes}

\definecolor{seq0}{HTML}{A2ABAB}
\definecolor{seq1}{HTML}{7D869C}
\definecolor{seq2}{HTML}{586994}

\newcommand{\myTitle}{Identifying Causal Effects in Information Provision Experiments}

\newcommand{\myName}{Dylan Balla-Elliott}

\newcommand{\thisMonth}{\ifcase \month \or January\or February\or March\or April\or May \or June\or July\or August\or September\or October\or November\or December\fi~ \number \year}

\hypersetup{
pdftitle = {\myTitle},
pdfauthor = {\myName}
}

\DefineBibliographyStrings{english}{
  backrefpage  = {p. },
  backrefpages = {p. },
}

\renewcommand\footnotesize{\small}

\available{pdfs.dballaelliott.com/info_iv.pdf}

\begin{document}
\addtocontents{toc}{\protect\setcounter{tocdepth}{-1}}

\setlength{\abovedisplayskip}{.25em}
\setlength{\belowdisplayskip}{.25em}

\title{\myTitle}
\author{\myName \thanks{\texttt{dbe@uchicago.edu} University of Chicago, Kenneth C. Griffin Department of Economics. \\
Thanks especially to Zo\"e Cullen, Ricardo Perez-Truglia,  Alex Torgovitsky, and Max Tabord-Meehan for early feedback and also to Magne Mogstad, Julia Gilman, Santiago Lacouture, Max Maydanchik, Isaac Norwich, Francesco Ruggieri, Sofia Shchukina, Alex Weinberg, Jun Wong, Itzhak Rasooly, Vod Vilfort, Whitney Zhang and many conference and seminar participants at the University of Chicago, UChicago Booth School of Business, and Purdue for helpful comments and suggestions. I am also indebted to Armona, Cantoni, Coibion, Fuster, Kumar, Gorodnichenko, Roth, Settele, Wiswall, Wohlfart, Yang, Yuchtman, Zafar, and Zhang for their useful replication packages.
    This material is based on work supported by the National Science Foundation Graduate Research Fellowship under Grant No. DGE 1746045. Any opinion, findings, and conclusions or recommendations expressed in this material are those of the authors and do not necessarily reflect the views of the National Science Foundation.
    A companion R package is available at \texttt{dballaelliott.github.io/lls/}.
}
}

\date{\thisMonth}

\begin{titlepage}\maketitle
Standard estimators in information provision experiments place more weight on individuals who update their beliefs more in response to new information.
This paper shows that, in practice, these individuals who update the most have the weakest causal effects of beliefs on outcomes.
Standard estimators therefore understate these causal effects.
I propose an alternative local least squares (LLS) estimator that recovers a representative unweighted average effect in a broad class of learning rate models that generalize Bayesian updating.
I reanalyze six published studies. In five, estimates of the causal effects of beliefs on outcomes increase; in two, they more than double.

\medskip
\noindent JEL CODES: C26, C9, D83, D9
 \end{titlepage}

\clearpage
\clearpage

\newcommand{\panelDate}{20250716}
\newcommand{\activeDate}{20250716}
\newcommand{\passiveDate}{20250716}

\newcommand{\Cov}{\cov}
\newcommand{\Var}{\var}

Information provision experiments have become a standard tool for studying the causal effects of beliefs \citep{wiswallDeterminantsCollege15, bottanChoosingYour22, jensenPerceivedReturns10}. But standard panel and two-stage least squares (TSLS) estimators systematically misrepresent average effects because they overweight individuals who update their beliefs the most. This matters because individuals whose beliefs most strongly affect their choices tend to update their beliefs the least, perhaps because they already sought out information before the experiment began. I propose a local linear slopes (LLS) estimator that weights all individuals equally. In five of six recent studies I reanalyze, LLS yields substantially larger estimates; in two cases the estimates more than double.

This paper is about experiments that study the causal effects of beliefs: how beliefs affect behavior, policy preferences, and even other beliefs.
In these experiments, researchers vary the information (\q{signal}) shown to participants, then estimate the effect of beliefs on behavior using panel or TSLS regressions. It is well known that such estimators target weighted averages of individual causal effects.\footnote{The weighted average interpretation of TSLS follows from \citet{imbensIdentificationEstimation94}. Similar results apply to difference-in-differences and other settings \citep{sunEstimatingDynamic20, goodman-baconDifferenceindifferencesVariation21, callawayDifferenceinDifferencesMultiple21}.}
In information provision experiments, these weights are proportional to the first-stage effect of information on beliefs.

Strong dependence between belief updating and belief effects makes panel and TSLS estimators substantially misrepresent average effects. When belief updating is negatively correlated with belief effects, standard panel or TSLS estimators can severely understate the average effect. The central empirical finding of this paper is that belief effects and belief updating are systematically negatively correlated: individuals whose beliefs most strongly affect their choices tend to update their beliefs least when provided new information.

I therefore propose a local least squares (LLS) estimator that consistently estimates an unweighted average effect, even when there is strong dependence between belief updating and belief effects. Researchers may prefer targeting an unweighted average as it is a representative summary of heterogeneous effects. \footnote{In an early application, \citet{guentherPoliticalRepresentation25} use results from the working paper version of this paper \citep{infoiv_v4}. They use LLS because unequal weights cause \q{2SLS [to] substantially misrepresent average effects} and so \q{we adopt [LLS] which identifies the unweighted average effect (p. 18-19)}}
This estimator can be applied to panel, active control, and passive control experiments.\footnote{
The LLS estimator applies immediately in the panel experiment. In experiments with active control groups, the LLS estimator identifies an unweighted average under a learning rate updating assumption. In experiments with passive control groups, the LLS estimator identifies the unweighted average when the variance of the prior is elicited in addition to the mean and the learning rate comes from Bayesian updating. An alternative approach with a passive control imposes the strong assumption that covariates are sufficiently rich to predict the belief update and that there is no residual variation in beliefs that cannot be predicted (i.e. \q{selection on observables}.)}

I apply the LLS estimator to six recent information provision studies published in leading economics journals.\footnote{
These applications span diverse contexts: college major choice \citep{wiswallDeterminantsCollege15}, housing investment \citep{armonaHomePrice19}, gender policy preferences \citep{setteleHowBeliefs22}, household \citep{rothHowExpectations20} and firm \citep{kumarEffectMacroeconomic23} responses to macroeconomic uncertainty, and protest participation \citep{cantoniProtestsStrategic19}. These six studies include examples of within-person panel experiments, and between person experiments with both active and passive control groups.
}
In five of these six applications, the LLS estimates are meaningfully larger than the panel or TSLS estimates. In two cases the estimates more than double.
To study mechanisms, I show how LLS can also be used to estimate effects of beliefs on outcomes conditional on the learning rate.
Empirically, belief effects are generally larger for the groups with smaller learning rates.
A simple model of endogenous information acquisition can rationalize this pattern. People whose beliefs strongly affect their decisions are incentived to form precise priors; when researchers provide new information, they update only modestly. When beliefs matter less, people start with noisier priors and update more.\footnote{\citet{mackowiakRationalInattentionForthcoming} consider a similar model with rational inattention before and during the experiment. Since they argue that the rational inattention dynamics \textit{before} the experiment dominate, their results are consistent with the model proposed in Appendix \ref{sec:endog_info_model} that does not include rational inattention during the experiment. }

The identification arguments in this paper use results in correlated random coefficients models from \citet{mastenIdentificationInstrumental16} and \citet{grahamIdentificationEstimation12}, generalized here to a nonparametric potential outcomes framework. \citet{vz_aeri} study TSLS in information provision experiments and provide conditions under which TSLS targets \textit{some non-negatively weighted average}. This paper proposes an alternative to TSLS that targets the \textit{equally-weighted average}.

The remainder of this paper is organized as follows.
Section \ref{sec:setup-notation} develops the conceptual framework.
Section \ref{sec:lit_weights} shows that standard panel and TSLS estimators target weighted averages of individual slopes; panel regressions have negative weights.
Section \ref{sec:lls} proposes the LLS estimator, which identifies an unweighted average.
Section \ref{sec:comparing_estimators} shows that under linearity, the unweighted average can be used to extrapolate.
Section \ref{sec:application} shows that attenuation is empirically widespread.
Section \ref{sec:conclusion} concludes.
 
\section{Conceptual Framework and Identifying Assumptions} \label{sec:setup-notation}

This paper is about experiments that study how beliefs affect behavior.
I analyze three leading experimental designs: panel experiments that compare the same individual before and after information provision, active control experiments that compare individuals receiving different signals, and passive control experiments that compare treated individuals to an untreated control group.\footnote{In between-subject experiments (with active or passive controls), I will focus on experimental designs where the information treatment is quantitative, for example \q{\textit{12 percent of the US population are immigrants}} \citep{hopkinsMutedConsequences19,grigorieffDoesInformation20} and not treatments that are qualitative, for example \textit{\q{[t]he chances of a poor kid staying poor as an adult are extremely large}} \citep{alesinaIntergenerationalMobility18}. The results for within-person (panel) experiments extend to qualitative or other kinds of signals.}

The identification argument follows a simple causal chain: treatment assignment $Z$ determines the signal $S$ shown to participants, which affects their beliefs $X$, which in turn affects outcomes $Y$. This $Z \rightarrow S \rightarrow X \rightarrow Y$ structure allows us to study how exogenous variation in information provision translates into belief changes and ultimately behavioral responses.
I formalize this causal chain in three parts: the outcome equation that links beliefs to behavior, the experimental designs that generate exogenous variation in beliefs, and the identifying assumptions that permit causal inference.

\subsection{Potential Outcomes}

The outcome equation allows for arbitrary heterogeneity in how beliefs affect outcomes:
\begin{align}
    Y_i &= G_i(X_i) \label{eq:outcome}
\end{align}
where $Y_i$ is the outcome or behavior of interest, $X_i$ is the belief, and $G_i(\cdot)$ is the individual-specific response function. The function $G_i(\cdot)$ generates potential outcomes: $Y_i(x) = G_i(x)$, where $Y_i(x)$ is $i$'s potential outcome when beliefs are exogenously set to $x$. This formulation places no restriction on  treatment effect heterogeneity; agents can differ both in their average responsiveness to beliefs and in the shape of their response functions.

We assume that beliefs $X_i$ are endogenous in the sense that $\E[Y_i \mid X_i = x] \neq \E[G_i(x)]$ for at least some $x$. This says that the difference in outcomes at two values of $X$ is not a causal effect.\footnote{If $G_i(x) = c x + U_i$ this is a familiar expression of endogeneity bias $\E[U_i \mid X_i] \neq \E[U_i]$.} This occurs when unobserved determinants of outcomes also affect beliefs.

\subsection{Experimental Designs} \label{eq:experiment_setup}
This paper considers three broad classes of information provision experiments. The first design uses within-person panel variation.

\begin{quote}

\textbf{Panel:}
The panel design uses contrasts within-individual before and after the information treatment. The {first-stage} variation in beliefs induced by treatment is the individual difference between beliefs before and after the information treatment.
\end{quote}

The second and third designs use between-person variation, but differ in the construction of the control group.

\begin{quote}
\textbf{Active Control:}
    The active control design uses contrasts between individuals who see a \q{high} signal and those who see a \q{low} signal. The {first-stage} variation in beliefs induced by treatment is the individual difference between potential beliefs if shown the \q{high} signal instead of the \q{low} signal.
\end{quote}

\begin{quote}
\textbf{Passive Control:}
The passive control design uses contrasts between individuals who recieve a signal and those who do not. The {first-stage} variation in beliefs induced by treatment is the individual difference between potential beliefs if shown the signal instead of not being shown the signal.
\end{quote}

Within and between person designs use different kinds of identifying variation and rely on qualitatively different kinds of identifying assumptions. The within person design uses the panel structure on the outcome and does not rely on any assumption on {how} people update beliefs in response to new information. In contrast, between person designs use assumptions on belief updating to match treatment units to the appropriate control units.\footnote{
In principle, panel and active control designs could be combined by eliciting pre-treatment outcomes in an active control experiment. Exploring the identification implications of such hybrid designs is beyond the scope of this paper but is an interesting direction for future research.
}

\subsection{Panel Identifying Assumption}

The identifying assumption is that outcomes follow a \q{panel} form. Let time $t$ have two periods, denoting pre ($t=0$) and post ($t=1$) information provision. Then, let
\begin{equation}
Y_{it} =  G_i(X_{it}) + \gamma_t  \label{eq:panel_outcome}
\end{equation}
The response function $G_i(\cdot)$ is time-invariant but arbitrarily heterogeneous across individuals; the time effects $\gamma_t$ are additively separable. This is a nonparametric generalization of the standard panel model used in the literature (e.g. \cite{wiswallDeterminantsCollege15,armonaHomePrice19}). The special case $G_i(x) = \tau_i x + U_i$ generates the classic linear panel model $Y_{it} =  \tau_i X_i + U_i + \gamma_t$ with heterogeneous treatment effects.

The identifying assumption that different changes in outcomes are due only to different changes in beliefs.\footnote{The time trend $\gamma_t$ is commonplace in empirical practice \citep{armonaHomePrice19,wiswallDeterminantsCollege15}. This allows for all respondents to, for example, respond with a higher number when the outcome is re-elicited, perhaps because of salience or other behavioral factors. The time trend $\gamma_t$ can be interacted with observables $W_i$ to allow for these time trends to vary across observables, like the prior belief. Without a time trend, the model implies that outcomes should not change when beliefs do not change: $\E \bs{\Delta Y_i | \Delta X_i = 0} = 0$. This restriction is testable in the data.}
There are no assumptions on how beliefs are updated; researchers who do not wish to place structure on belief updating may find the panel design particularly appealing.

\subsection{Active and Passive Control Identifying Assumption}

In active and passive control experiments, the relationship between beliefs and outcomes is completely flexible. The identifying assumption is that belief updating follows a simple {learning rate} structure. This includes the workhorse linear updating or \q{signal averaging} models like Bayesian updating. Randomization to a particular signal generates variation in posterior beliefs through this learning rate updating.

\subsubsection{Learning Rate Belief Updating} \label{sec_setup_updating}

For exposition, the main text uses the familiar linear form throughout; potential beliefs are a linear function of the prior $X_i^0$ and an experimental signal $s$:
\begin{align}
        X_i(s) & = \alpha_i \bp{s - X_i^0} + X^0_i \label{eq:beliefs_po_signal}
\end{align}
The heterogeneous coefficient on the signal $\alpha_i$ is often called the learning rate. In this model, posterior beliefs are a weighted average of the prior and the signal, with weight $\alpha_i$ on the signal. This updating rule is widely used in applied work and fits observed belief changes well in information provision experiments.\footnote{See for example \citet{cavalloInflationExpectations17, cullenHowMuch22, giaccobassoWhereMy22, cullenIncreasingDemand23,fusterExpectationsEndogenous22}.}

This linear updating rule is often microfounded in a normal-normal Bayesian updating, but it also arises in several other behavioral models.This class of linear updating models includes rational inattention \citep{fusterExpectationsEndogenous22}, base-rate neglect, over-reaction, under-reaction \citep{gretherBayesRule80} and anchoring on the prior or signal \citep{gabaixBehavioralInattention19}.
\footnote{
Linearity in belief updating can be relaxed as long as differences in updating are still driven only by the learning rate. Nonlinear learning rate models take the form $X_i(s) = \alpha_i f(s, X_i^0) + X^0_i$, where $f(\cdot, X_i^0)$ is any function monotonic in the signal with $f(X_i^0, X_i^0) = 0$. For example, $f$ could be a nonlinear \q{dampener} that discounts signals further away from the prior. Or, it could be asymmetric around zero so that people respond more to signals of a particular sign.
The remainder of the paper uses the linear updating rule with $f(s, X_i^0) = s-X_i^0$ due to its overwhelming popularity in practice and because it can be microfounded in many popular models of belief updating.}
See Appendix \ref{sec:learning_rate_models} for further discussion.

\subsubsection{Randomization and Potential Beliefs}

Denote treatment arms by $Z_i$. In the active and passive control designs, assume that the researcher randomizes over two arms $Z_i \in \{A,B\}$. In the active design, arm $A$ will be the treatment arm that receives the \q{high} signal and arm $B$ will be the treatment arm that receives the \q{low} signal. In the passive design, arm $A$ will be the treatment arm that receives a signal and arm $B$ will be the control arm that does not receive a signal.
Finally, $S_i(z)$ is the signal that is shown to individual $i$ in treatment arm $z$.\footnote{In the panel design, the researcher may randomly assign $Z_i$ in the same way, or may chose to show the information to all participants. If the panel design includes a treatment arm that receives no information, denote that arm with $B$. Since the panel design uses within-person contrasts, identification does not come from randomization across people. Thus it is sufficient to work with the realized signal $S_i$.}

Treatment is assigned randomly in the sense that $Z_i$ is independent of the potential outcomes: the outcome function $G_i(\cdot)$, the prior $X^0_i$, the potential signals $S_i(\cdot)$, and the learning rate $\alpha_i$.
\footnote{While the treatment $Z_i$ will be randomly assigned, it is important to note that the realized signal $S_i(Z_i)$ can generally vary across individuals endogenously.
In \citet{bottanBettingHouse22}, $S_i(A)$ and $S_i(B)$ are high and low estimates of the home value and thus the realized signal is only randomly assigned conditional on the potential signals.}
In passive designs, treatment arm $B$ does not receive any signal. For the sake of completeness, define $S_i(B) \equiv X^0_i$ in passive designs. It will be convenient to work with the following shorthand where potential beliefs are directly a function of the treatment assignment $z$. In a slight abuse of notation, we redefine
\begin{align}
X_i(z) &\equiv   X_i(S_i(z)) = \alpha_i \bp{S_i(z) - X_i^0} + X^0_i
 \label{eq:beliefs_PO}
\end{align}

We will use this equation for potential beliefs along with the potential outcome equations \eqref{eq:outcome} and \eqref{eq:panel_outcome} to study common empirical specifications.

\section{Standard Panel and TSLS Estimators} \label{sec:lit_weights}

The following three sections compare standard estimators to a local least squares (LLS) alternative. Standard estimators weight individuals by their belief updating; LLS weights all individuals equally. When belief updates are negatively correlated with causal effects, standard estimators understate the average effect.
The current section begins by introducing the individual slopes, which are the causal building block of all the estimators considered in this paper, and then shows that standard panel and TSLS estimators recover weighted averages of these slopes.

\subsection{Individual Slopes: The Causal Building Block}

Define the individual slope as the ratio of outcome change to belief change induced by the experiment:
\begin{equation}
	\beta_i \equiv \frac{G_i(X_i(A)) - G_i(X_i(B))}{X_i(A) - X_i(B)} \label{eq:defn_indiv_slopes}
\end{equation}
This is the average rate of change in individual $i$'s outcome as beliefs move from $X_i(B)$ to $X_i(A)$, which are the individual-specific beliefs in treatment arms $A$ and $B$. Equivalently, this is the individual-specific average partial effect $G'_i(x)$ over the individual-specific interval of beliefs induced by the experiment.
These individual slopes $\beta_i$ thus depend both on the individual response function $G_i(\cdot)$ and the variation in beliefs induced by the experiment
$\bc{X_i(B), X_i(A)}$. In the panel design, define $X_i(A) \equiv X_{i1}$ and $X_i(B) \equiv X_{i0}$.
The standard estimators used in the literature and the new LLS estimator aggregate these individual slopes differently. The differences between the parameters targeted by LLS and TSLS or panel estimators come \textit{entirely} from differences in aggregation. The remainder of this section characterizes standard panel and TSLS estimators.

\subsection{Standard Panel and TSLS Specifications}

Standard estimators in information provision experiments yield weighted averages of individual effects $\beta_i$, with weights proportional to belief updating. In panels, individuals with below-average belief updates receive negative weights.

\begin{equation}
    \beta^{design} \equiv \E\bs{ \beta_i \times \omega_i(design)} \label{eqn:standard_causal_param}
\end{equation}

The precise form of these weights varies, but in all three cases, standard specifications weight individual effects $\beta_i$ in proportion to the first-stage belief updating. In all specifications, these weights integrate to one. Appendix \ref{sec:weight_derivations} contains derivations for all expressions in this section and Appendix \ref{app_tsls} provides a more general discussion of TSLS in information experiments.

We now examine three representative specifications and derive the implicit weights each places on different individuals.

\subsubsection{A Representative Panel Specification} \label{sec_panel_fd_reg}

\citet{armonaHomePrice19} use a regression in first-differences. Since there are only two time periods, this is equivalent to a panel regression with individual and time fixed effects. Let $\Delta X_i$ denote the difference between the post- and pre-treatment observations, $X_{i1} - X_{i0}$.  The regression specification is simply
\begin{align}
    \beta^{Panel} &\equiv \frac{\cov \bs{\Delta Y_i,  \Delta X_i}}{\var \bs{\Delta X_i}} \label{eq:panel_regression}
    \intertext{which has implied weights}
    \omega_i(Panel) &\propto \Delta X_i(\Delta X_i - \E[\Delta X_i]) \label{eq:panel_weights}
\end{align}

The regression of $\Delta Y_i$ on $\Delta X_i$ and a constant can assign negative weights to observations with $\Delta X_i$ between zero and the mean $\E[\Delta X_i]$.

\paragraph{Heterogeneity Bias Causes Negative Weights in Panel Regressions}

This negative weights result restates \citeauthor{chamberlainMultivariateRegression82}'s classic (\citeyear{chamberlainMultivariateRegression82})  \q{heterogeneity bias} as negative weights in a weighted average of individual effects. A closely related expression appears in Theorem 3.4c of \citet{callawayDifferenceinDifferencesContinuous25}, who show that units with below-mean treatment intensity receive negative weights in difference-in-differences with continuous treatment. The panel regression here is analogous: it compares outcomes for big changers to small changers. Small changers act as the control group and their outcomes are subtracted from outcomes for big changers. Increasing the treatment effects of small changers thus \textit{decreases} the slope estimate. This is what is means for them to have negative weights. Heterogeneity bias arises because these cross-update comparisons are contaminated by differences in treatment effects.

\subsubsection{A Representative Active Control Specification}

   \citet{setteleHowBeliefs22} uses an IV specification where assignment to the \q{high} signal $T_i \equiv \1\bc{Z_i = A}$ is a binary instrument for beliefs. The estimand takes the canonical Wald form:

   \begin{align}
    \beta^{Active} &\equiv \frac{\E \bs{ Y \mid Z = A} - \E \bs{ Y \mid Z = B}}{\E \bs{ X \mid Z = A} - \E \bs{ X \mid Z = B}} \label{eq:active_regression} \\
    \omega_i(Active) &\propto X_i(A) - X_i(B)
    \intertext{which under learning rate updating simplifies further to}
    \omega_i(Active) &\propto \alpha_i (S_i(A) - S_i(B)) \label{eq:active_weights}
\end{align}

These weights are non-negative under learning rate updating with $\alpha_i \geq 0$ and in a general class of updating models when a monotonicity assumption holds such that $(X_i(A)-X_i(B))$ has the same sign for everyone.

\subsubsection{A Representative Passive Control Specification} \label{sec:main_passive_exp_weights}

\citet{cullenIncreasingDemand23} use an IV specification where the instrument is an indictor for assignment to the information treatment interacted with the initial gap in beliefs.\footnote{\citet{vz_aeri} point out that similar specifications that also include the treatment indictor as an excluded instrument have negative weights.}
\begin{equation}
T^{ex}_i \equiv T_i(S_i(A) - X_i^0)
\end{equation}
Since these specifications control for the exposure $S_i(A) - X_i^0$, the residual variation in the instrument is simply a re-centered version of the instrument.\footnote{To see this, notice that random assignment implies that
$\E \bs{T^{ex}_i \mid S_i(A) - X_i^0} = \E \bs{T_i} \bp{S_i(A) - X_i^0} = \L \bs{T^{ex}_i \mid S_i(A) - X_i^0}$. By FWL $\wt{T}^{ex}_i  \equiv T^{ex}_i - \L \bs{T^{ex}_i \mid S_i(A) - X_i^0} = (T_i - \E[T_i])(S_i(A) - X_i^0) $. }
\begin{equation}
\wt{T}^{ex}_i \equiv (T_i - \E[T_i])(S_i(A) - X_i^0)
\end{equation}
The TSLS coefficient is then given by
\begin{align}
    \beta^{Passive}  &\equiv \frac{\cov \bs{\wt{T}^{ex}_i, Y_i}}{\cov \bs{\wt{T}^{ex}_i, X_i}} \label{eq:exposure_reg} \\
    \omega_i(Passive) &\propto (X_i(A) - X_i(B)) (S_i(A)- X_i^0) \label{eq:exposure_reg_base_wt} \\
    \intertext{which under learning rate updating simplifies further to}
    \omega_i(Passive) &\propto \alpha_i (S_i(A)-X^0_i)^2 \label{eq:exposure_reg_simple_wt}
\end{align}

These weights are non-negative under learning rate updating with $\alpha_i \geq 0$ and in a general class of updating models when monotonicity holds: $\text{sign}(X_i(A)-X_i(B)) = \text{sign}(S_i(A)-X_i^0)$.

\subsection{Discussion}

The key takeaway from these expressions is that these standard specifications weight individual effects by the strength of belief updating. In the active and passive controls, weights are non-negative and thus are \q{weakly causal}.

 \section{The Local Least Squares Estimator} \label{sec:lls}

This section presents a local least squares (LLS) estimator that recovers an equally weighted average of individual belief effects. With a linear outcome equation, these individual slopes have a structural interpretation as partial derivatives of the outcome with respect to beliefs and so the equally weighted average is the (structural) average partial effect (APE).

LLS is a control function estimator. It works by constructing a vector of controls that isolates the experimental variation in beliefs. In this setting, learning rate updating means that people who have the same prior, the same potential signals, \textit{and the same learning rate} have the same potential beliefs; the only variation in their actual beliefs comes from the random assignment to the actual signal.
The LLS approach aggregates many \q{local} regressions that use only this exogenous (i.e. experimental) variation in beliefs.\footnote{\citet{mastenIdentificationInstrumental16, grahamIdentificationEstimation12} show how to construct these \q{local} regressions in panel and IV settings more generally. I generalize their results from the linear random coefficients model to a more general nonparametric potential outcome model.}

\subsection{Intuition: Conditioning on Potential Beliefs} \label{sec:ape_identification}

The LLS estimator recovers equally weighted averages of individual slopes $\E \bs{\beta_i}$ by constructing local regressions that isolate purely experimental variation in beliefs. The ideal regression conditions on the potential beliefs $X_i(A)$ and $X_i(B)$, which isolates only the remaining variation in beliefs that comes from being assigned randomly to treatment $A$ or $B$. This ideal regression is:

\begin{equation*}
    \frac{\cov \bs{Y_i, X_i \mid X_i(A) = x_A, X_i(B) = x_B}}{\var \bs{X_i \mid X_i(A) = x_A, X_i(B) = x_B}} = \E \bs[\bigg]{\underbrace{\frac{G_i(X_i(A)) - G_i(X_i(B))}{X_i(A) - X_i(B)}}_{\equiv \beta_i} \mid X_i(A) = x_A, X_i(B) = x_B }
\end{equation*}

This regression recovers a conditional average $\E \bs[\big]{\beta_i \mid X_i(A) = x_A, X_i(B) = x_B }$. Iterating expectations thus recovers the average individual slope $\E[\beta_i]$. This is an easily interpretable causal parameter: it answers the question, ``On average, how much do outcomes change per unit change in beliefs, over the range of beliefs induced by the experiment?''.

The LLS estimation strategy also produces intermediate estimates $\E[\beta_i \mid \alpha_i]$ that reveal how causal effects vary with belief updating.
Many behavioral models make strong predictions about the relationship between belief updating and belief effects \citep{mackowiakRationalInattentionForthcoming, yangDecisionRelevanceSubjective24, enkeBehavioralAttenuation24, fusterExpectationsEndogenous22}.
Section \ref{sec:application} presents estimates of these conditional average slopes to document strong negative correlation between belief updates and causal effects across a range of settings.

The identification strategy in practice is then to condition on a set of controls that is as good as conditioning on the potential beliefs directly. The following section shows how to construct feasible local regressions.

\subsection{Constructing Feasible Local Regressions}

The following sections show to construct feasible local regressions in the three experimental designs. Appendix \ref{sec:identification_derivation} provides proofs for the results in this section.

\subsubsection{Local Regressions in Panel Experiments} \label{main_panel_APE_id}
The panel approach works with any information treatment (including qualitative treatments or bundles of signals) because identification relies only on the panel structure, not on the content of the signal.

For any belief change $x \neq 0$:
\begin{equation}
\E \bs{\beta_i \mid \Delta X_i = x } = \frac{\cov \bs{\Delta Y_i, \Delta X_i \mid \Delta X_i \in \bc{0, x}}}{\var \bs{\Delta X_i \mid \Delta X_i \in \bc{0, x}}} \label{eq:panel_LLS}
\end{equation}
The right hand side is a feasible local regression using only observations with $\Delta X_i = x$ or $\Delta X_i = 0$. Iterating over $x$ and averaging yields $\E[\beta_i]$. This requires that some individuals have (close to) zero change in beliefs.\footnote{This is an easily verifiable condition. It is satisfied if $P\bs{\Delta X_i = 0} > 0$, or more generally if $\Delta X_i$ has positive mass in any neighborhood around zero. See \citet{grahamIdentificationEstimation12} for detailed discussion of technical considerations with continuous $\Delta X_i$.}

\subsubsection{Local Regressions in Active Control Experiments}
Active designs rely on the Bayesian updating assumption \eqref{eq:beliefs_PO} and identify learning rates directly from observed belief updates: $\alpha_i = \bp{X_i - X_i^0}/\bp{S_i - X_i^0}$. Under Bayesian updating, people with the same learning rate, prior, and potential signals have the same potential beliefs; the only remaining variation comes from random assignment.

The control vector is $C_i \equiv \bs{\alpha_i \; X_i^0 \; S_i(A)\;  S_i(B)}$. Conditional on $C_i = c$:
\begin{equation}
\E \bs{\beta_i \mid C_i = c } = \frac{\cov \bs{Y_i, X_i \mid C_i = c}}{\var \bs{X_i \mid C_i = c}}
\label{eq:main_active_LLS_id_linear}
\end{equation}
Iterating over $c$ and averaging yields $\E[\beta_i]$. The regression is feasible when $(S_i - X_i^0) \neq 0$ and $\var\bs{X_i \mid C_i = c} > 0$, which excludes cases with no learning ($\alpha_i = 0$) or identical signals ($S_i(A) = S_i(B)$)

\subsubsection{Local Regressions in Passive Control Experiments} \label{sec:lls_passive}

Passive designs also rely on the Bayesian updating assumption \eqref{eq:beliefs_PO}, but require additional assumptions because learning rates for the control group are unobserved. Consider two possible approaches to infer learning rates in the control group:

\paragraph{Case 1: Observed Prior Variance} In normal-normal Bayesian updating, $\alpha_i = {{\sigma^2_X}_i}/ \bp{{\sigma^2_X}_i  + \sigma^2_S}$. If signal precision $\sigma^2_S$ is common across individuals, then conditioning on the rank of prior variance ${\sigma^2_X}_i$ is equivalent to conditioning on $\alpha_i$. The control vector becomes $C_i \equiv \bs{\text{rank}\bp{ {\sigma^2_X}_i} \; X_i^0 \; S_i(A)}$.

\paragraph{Case 2: Rich Observables} When researchers can predict beliefs from observables \citep{ballaelliott22,cantoniProtestsStrategic19}, they can use \textit{predicted updates} instead of observed updates. The implied predicted learning rate $\wt{\alpha}_i$ replaces the observed rate. The control vector becomes $C_i \equiv \bs{\wt{\alpha}_i \; X_i^0 \; S_i(A)}$.\par

In either case, under the linear outcome equation \eqref{eq:outcome} and Bayesian updating \eqref{eq:beliefs_PO}:
\begin{equation}
\E \bs{\beta_i \mid C_i = c } \equiv \frac{\cov \bs{Y_i, X_i \mid C_i = c}}{\var \bs{X_i \mid C_i = c}}
\end{equation}

Appendix \ref{app_ID_passive_APE} formally states the assumptions in both of these cases.

\subsubsection{Comparing Assumptions Across Designs} \label{sec_identifying_assns}

 The three experimental designs require progressively stronger assumptions to implement LLS. Panel designs impose no new behavioral assumptions. Active designs require Bayesian updating. Passive designs require Bayesian updating and also require either elicited prior variances or rich observables to infer unobserved learning rates.

The assumptions in the active case are weaker than in the passive case because in the active case researchers observe all participants update beliefs in response to new information. The experiment {reveals} heterogeneity in belief updating. In contrast, in a passive design, researchers need to use observables to {infer} heterogeneity in belief updating for a control group that the researcher never sees update their beliefs.\footnote{Recall that the learning rate is identified from the observed update $\alpha_i = \bp{X_i - X_i^0} / \bp{S_i(Z_i) - X_i^0}$, which is undefined for the passive control group that receives no information. Randomization is enough to ensure that the learning rates have the same distribution in both groups, but the individual learning rates are not directly identified in the passive control group.
} This suggests that researchers interested in implementing an LLS estimator may find active designs more attractive since they reveal more information about belief updating.\footnote{There are many design considerations beyond the scope of this paper. \citet{haalandDesigningInformation23} discuss implementation considerations of active and passive control designs. \citet{listExperimentalistLooks25} discusses within- and between-subject experimental designs more generally.}

\subsection{Practical Implementation}

Conditioning on high-dimensional control vectors is often impractical in experimental samples. When belief updating is linear in the signal and prior, it is sufficient to control for $C_i$ semi-parametrically.
The local regressions in between-person designs need only condition on the learning rate and can simply control linearly for the prior and signals in each local regression. In passive designs, or designs with person-specific high and low signals (i.e. \citet{rothRiskExposure22}), it is also necessary to reweight by the inverse of the exposure.
This weighted local regression recovers $\E \bs{\beta_i \mid \alpha_i}$. Appendix \ref{app_linear_controls_LLS} shows that this modified local regression is sufficient and Appendix \ref{sec:estimation} provides general implementation guidance.
 \section{Comparing Estimators and Interpreting Individual Slopes} \label{sec:comparing_estimators}

The estimators in Sections \ref{sec:lit_weights} and \ref{sec:lls} target parameters that can be written as $\E[\beta_i \times \omega_i]$ for some weights $\omega_i$. The interpretation of these parameters depends on the interpretation of the individual slopes $\beta_i$, but the difference between estimators comes only from the weights $\omega_i$. LLS assigns equal weights. Under linearity, these equal weights deliver the APE, which has a structural interpretation that permits extrapolation. Appendix \ref{app:nonlinearity_details} discusses the nonlinear case in greater detail.

\subsection{Equal Weights Deliver a Representative Average} \label{sec:equal_weights}

With treatment effect heterogeneity, researchers must decide how to summarize heterogeneous effects. LLS recovers a simple average $\E[\beta_i]$. This equally weighted average $\E[\beta_i]$ answers the question: \q{On average, how much did outcomes change per unit change in beliefs?} Like the non-parametric ATE $\E[G'_i(X_i)]$, this parameter is local to the variation the experiment actually induced \citep{heckmanChapter7107a}.
A TSLS-weighted average may be policy-relevant when the intervention under consideration is information provision, since it captures effects among those whose beliefs would actually change.\footnote{If the policy question is whether to implement an information campaign, the reduced form (the effect of treatment assignment on outcomes) answers this directly.} However, the attenuation documented in Section \ref{sec:application} suggests that relying on TSLS outside this narrow case is risky: researchers may conclude that belief effects are generally unimportant on the basis of an unrepresentative average.

\subsection{Under Linearity, the Average Slope Permits Extrapolation} \label{sec:linearity_ate}

In the linear outcome equation $G_i(x) = \tau_i x + U_i$, the parameter $\tau_i$ is structural: it fully characterizes $i$'s response to any hypothetical belief shift, not just those induced by the experiment. The average $\E[\tau_i]$ inherits this structural property: on average, a one-unit increase in beliefs causes an $\E[\tau_i]$-unit increase in outcomes, regardless of initial belief levels. This permits extrapolation; predictions for hypothetical interventions that shift beliefs by any amount can be formed by scaling the average effect appropriately.\footnote{Since $\tau_i$ is the individual partial effect, the average $\E[\tau_i]$ is also called the average partial effect (APE).}

\section{Empirical Applications} \label{sec:application}

This section demonstrates that attenuation due to dependence between belief updating and belief effect is empirically relevant. I compare standard panel and TSLS specifications to LLS estimates in six recent studies from leading economics journals .\footnote{
I searched the Web of Science database for papers in the top five economics journals, ReStat, AER: Insights, and all AEJs containing ``beliefs,'' ``information,'' or ``perception'' together with ``experiment'' or ``treatment.'' This yielded 116 potentially eligible experiments. I replicated the two most highly cited studies of each experimental design. To standardize the presentation of the results, I flip the sign of the outcome variable when necessary to ensure that mean effects are always positive. I also omit additional demographic controls and probability weights from all estimates for simplicity.}
See Appendix \ref{sec:estimation} for estimation details.

Table~\ref{tab:all_applications} contrasts LLS estimate with estimates recovered by the standard specification in each study. In five of the six studies, standard estimators are substantially attenuated.
Figure~\ref{fig:cape} plots an estimate of conditional slopes for each study: $\E \bs{\beta_i \mid \abs{\Delta X_i}}$ in panel experiments (Panel A) and $\E \bs{\beta_i \mid \text{rank}(\alpha_i)}$ in active and passive control experiments (Panels B and C). These curves directly that people with the strongest causal effects tend to have smaller belief updates.

\subsection{Results from Panel Experiments}

\citet{wiswallDeterminantsCollege15} study how beliefs about field-specific earnings affect college students' major choices.
The panel estimate of 0.32  (s.e. 0.086) is substantially smaller than the LLS estimate of 0.721 (s.e. 0.33), with the LLS estimate being 125\% larger.
\citet{armonaHomePrice19} study how beliefs about home prices affect investment decisions.
The panel estimate of 1.15  (s.e. 0.234) is smaller than the LLS estimate of  1.8 (s.e. 0.381), with the LLS estimate being over 50\% larger.

\subsection{Results from Active Control Experiments}

\citet{setteleHowBeliefs22} studies how beliefs about the gender wage gap affect support for gender equality policies. The TSLS estimate of 0.096 (s.e. 0.033) is substantially smaller than the LLS estimate of 0.16 (s.e. 0.042), with the LLS estimate being 66\% larger.
\citet{rothRiskExposure22} study how recession expectations affect subjective personal unemployment risk. Their TSLS estimate of 0.755 (s.e. 0.433) is somewhat smaller than the LLS estimate of 0.882 (s.e. 0.379), with the LLS estimate being 17\% larger.

\subsection{Results from Passive Control Experiments} \label{sec:applications_passive}

\citet{kumarEffectMacroeconomic23} study how beliefs about GDP growth affect employment decisions. The TSLS estimate 0.466 (s.e. 0.19) is smaller than the LLS estimate 1.787 (s.e. 0.409), with the LLS estimate being 284\% larger.
\citet{cantoniProtestsStrategic19} study how beliefs about others' protest participation affect one's own willingness to participate. The TSLS estimate (0.68, s.e. 0.253) and the LLS estimate (0.18, s.e. 0.133) are both quite noisy, making it difficult to draw strong conclusions about the direction or magnitude of any difference. The difference between the TSLS and LLS estimates is suggestive evidence that people with larger belief effects had \textit{larger} belief updates. However, the conditional effects in Panel C.ii of Figure \ref{fig:cape} reveals only modest variation across learning rate ranks, with quite wide confidence intervals.\footnote{Concerns of attenuation are only one reason among many to consider using the LLS estimator. The estimator consistently recovers the unweighted average regardless of the sign of dependence between belief updating and treatment effects. The pattern of attenuation observed in five of six applications is an empirical finding, not a mechanical feature of the estimator.}

\subsection{Discussion} \label{sec:applications_discussion}

The conditional effects in Figure \ref{fig:cape} reveal that individuals who update their beliefs the least have the strongest causal effects across many applications. This provides direct empirical support for models of endogenous information acquisition where people with decision-relevant beliefs invest in forming precise priors.
Indeed, the mechanism where people are well informed about things that matter for their decisions and so respond less to new information is quite general and reflects basic features of rational inattention (Appendix \ref{sec:endog_info_model}, see also \citet{mackowiakRationalInattentionForthcoming, fusterExpectationsEndogenous22, cavalloInflationExpectations17}).
 \section{Conclusion}
\label{sec:conclusion}

Standard empirical specifications in information provision experiments systematically understate the causal effects of beliefs on behavior. This paper demonstrates that in five of six high-profile studies in leading economics journals, ranging from college major choice to macroeconomic expectations, LLS estimates average effects of beliefs that are larger than estimates from standard specifications.

 \clearpage
\begin{table}[tbp]
\centering
\caption{LLS and Standard Specifications in Six Studies}
{\small
\begin{tabular}{lcc}
\toprule
\textbf{Panel A: Panel Experiments}  & LLS & FD Regression \\
\hspace{1em} Wiswall and Zafar (2015) & 0.721 & 0.320 \\
& (0.290) & (0.086)
  \\ 
\hspace{1em} Armona, Fuster, and Zafar (2019) & 1.800 & 1.147 \\
& (0.387) & (0.234)
  \\[-0.2em]
\midrule
\textbf{Panel B: Active Experiments}  & LLS & TSLS Regression \\
\hspace{1em} Settele (2022) & 0.160 & 0.096 \\
& (0.042) & (0.033)
  \\ 
\hspace{1em} Roth, Settele, and Wohlfart (2022) & 0.882 & 0.755 \\
& (0.365) & (0.435)
  \\[-0.2em]
\midrule
\textbf{Panel C: Passive Experiments}  & LLS & TSLS Regression \\
\hspace{1em} Kumar, Gorodnichenko, and Coibion (2023) & 0.100 & 0.023 \\
& (0.087) & (0.037)
  \\ 
\hspace{1em} Cantoni, Yang, Yuchtman, and Zhang (2019) & 0.180 & 0.680 \\
& (0.133) & (0.253)
  \\
\bottomrule
\end{tabular}
\label{tab:all_applications}
}
\begin{quote} \footnotesize \textit{Notes:}  \textit{Notes:} This table compares local least squares (LLS) estimates of the unweighted average effect to standard first-difference (FD) or two-stage least squares (TSLS) estimates across all six replication studies. 
Bootstrap standard errors are reported in parentheses. Appendix \ref{sec:estimation} discusses implementation details and reports results for alternative choices of bandwidth.
\end{quote}
\end{table}

\begin{figure}
    \centering
    \caption{Dependence between Belief Updating and Belief Effects in Six Studies}
    \label{fig_capes}
    { \small
    \textsc{Panel A:} Panel Experiments \\[.5em]
    \begin{minipage}{.5\textwidth}
    \centering I: \citet{wiswallDeterminantsCollege15} \\
            \includegraphics[width=\linewidth]{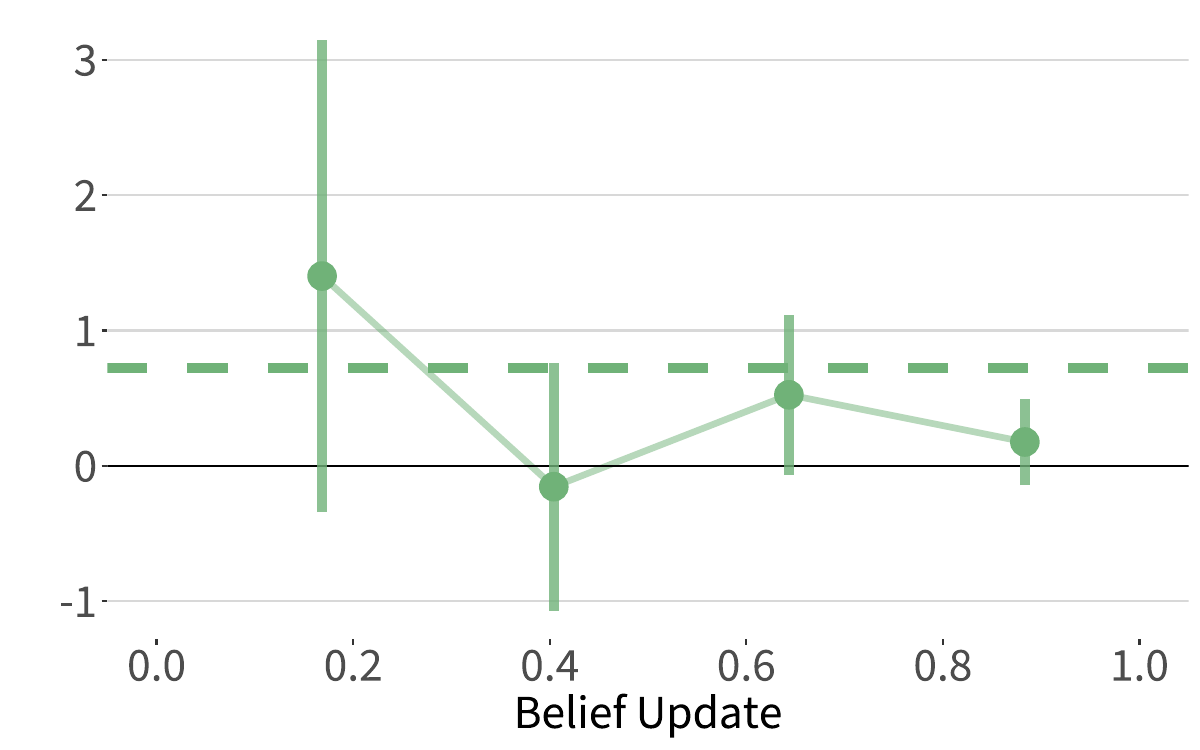}
    \end{minipage}\begin{minipage}{.5\textwidth}
    \centering II: \citet*{armonaHomePrice19} \\
            \includegraphics[width=\linewidth]{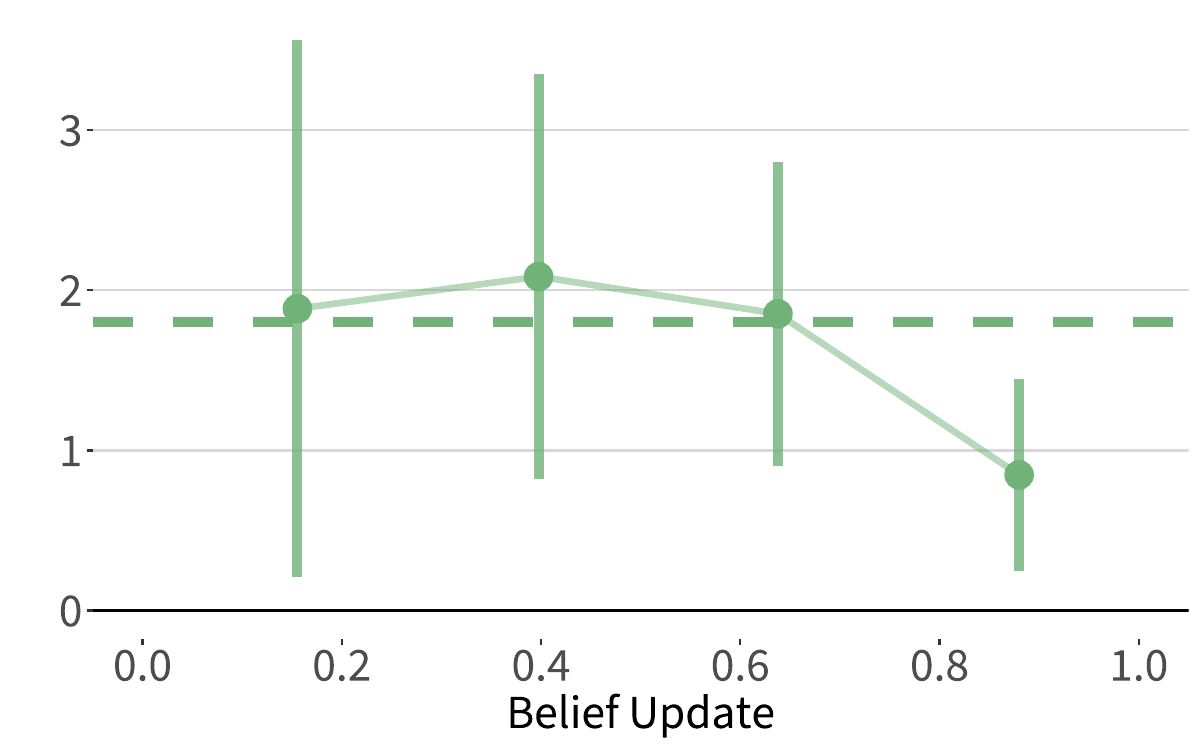}
    \end{minipage} \\
    \textsc{Panel B:} Active Experiments \\[.5em]
    \begin{minipage}{.5\textwidth}
    \centering  I: \citet{setteleHowBeliefs22} \\
            \includegraphics[width=\linewidth]{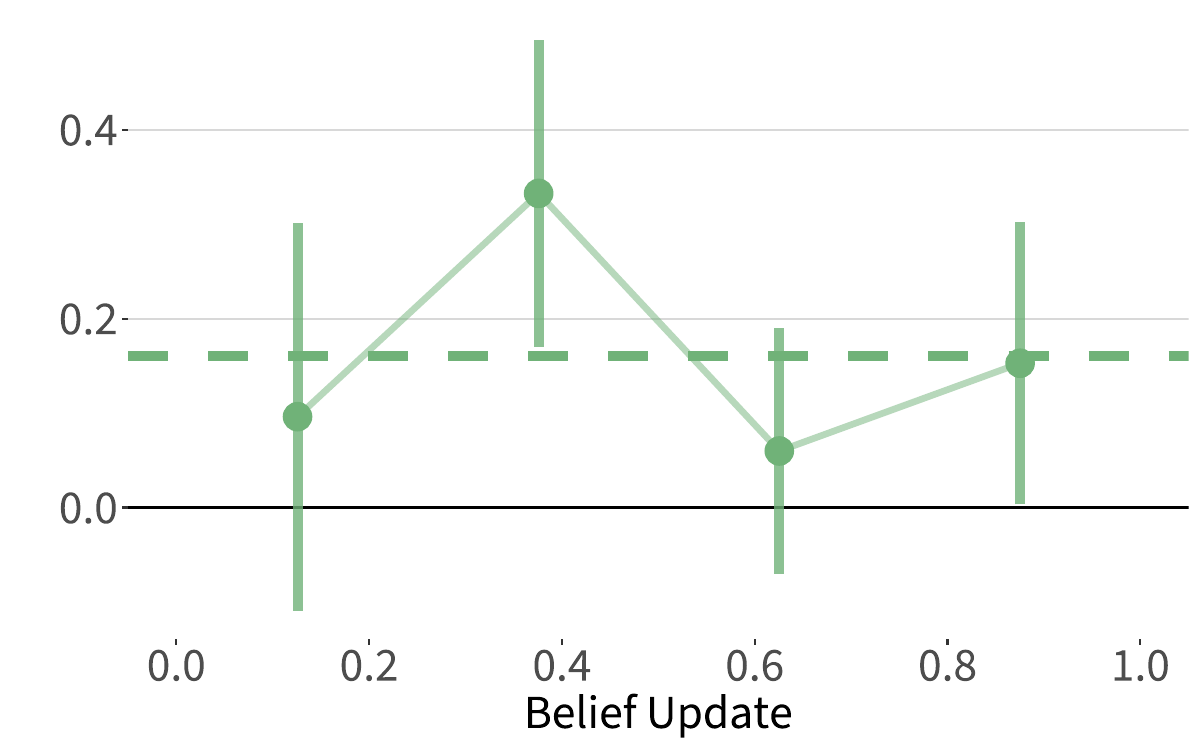}
    \end{minipage}\begin{minipage}{.5\textwidth}
    \centering  II: \citet{rothRiskExposure22} \\
            \includegraphics[width=\linewidth]{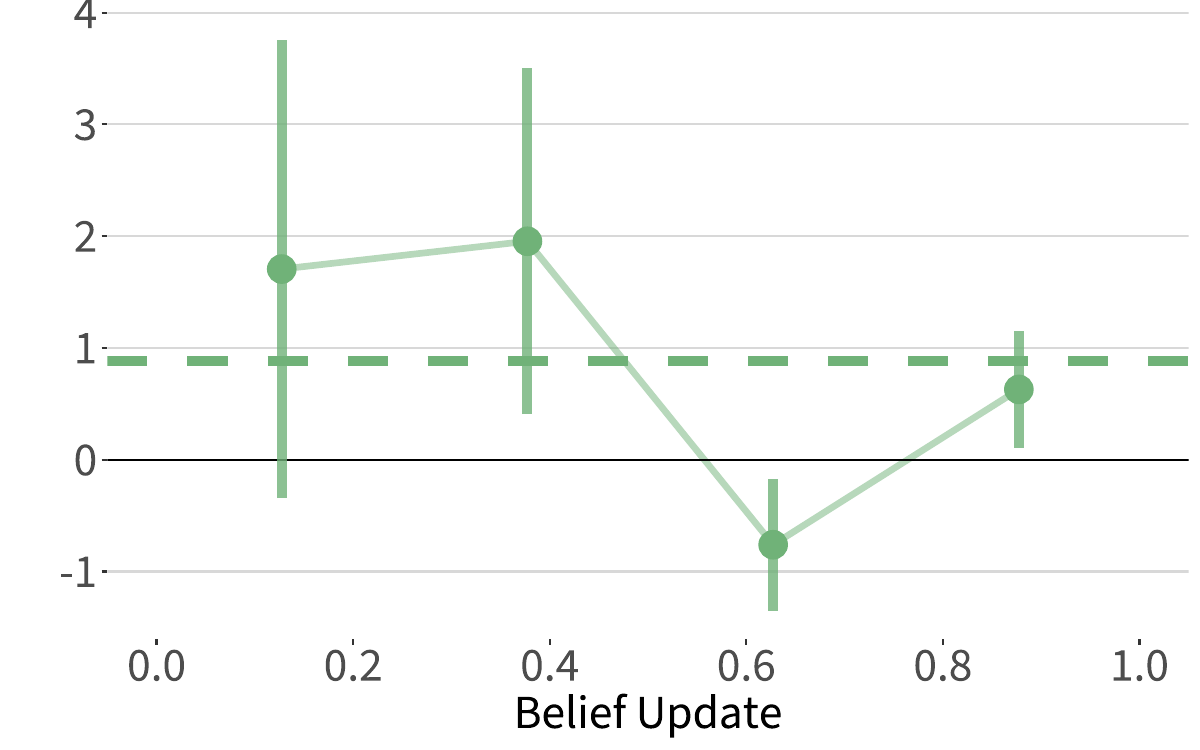}
    \end{minipage} \\
        \textsc{Panel C:} Passive Experiments \\[.5em]
    \begin{minipage}{.5\textwidth}
    \centering  I: \citet*{kumarEffectMacroeconomic23} \\
            \includegraphics[width=\linewidth]{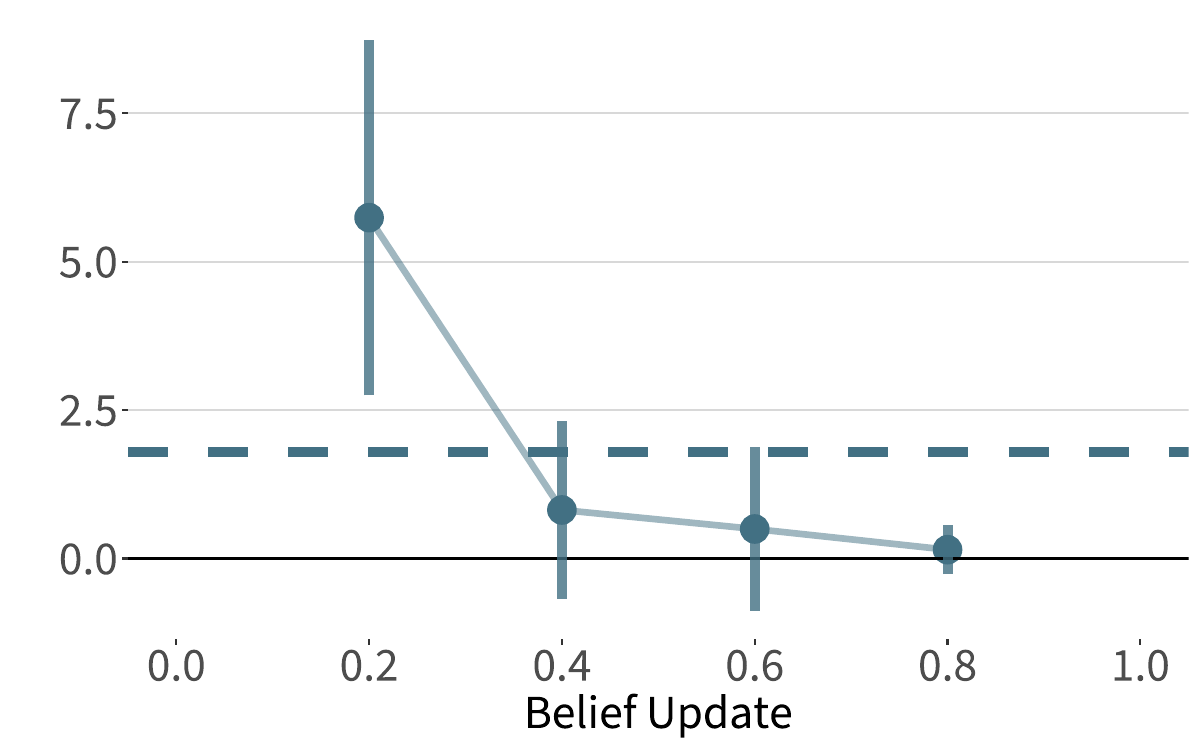}
    \end{minipage}\begin{minipage}{.5\textwidth}
    \centering  II: \citet*{cantoniProtestsStrategic19} \\
            \includegraphics[width=\linewidth]{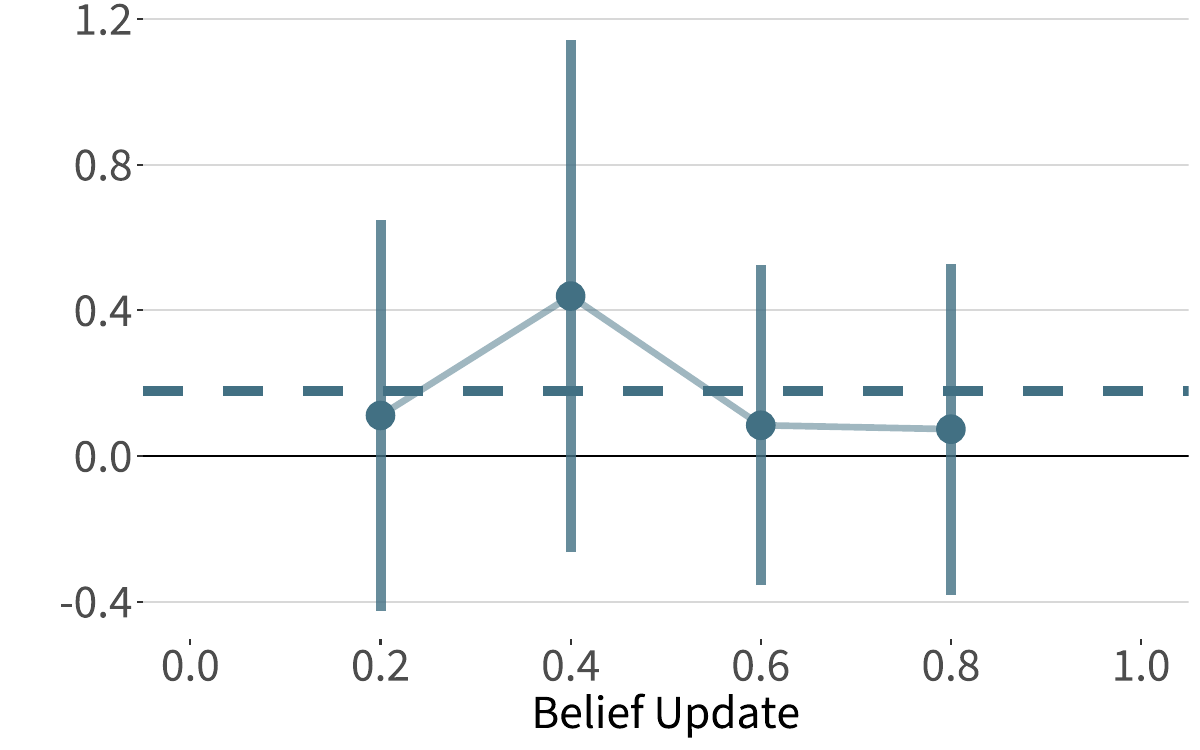}
    \end{minipage}
    \label{fig:cape}
    }
    \begin{quote} \footnotesize 
\textit{Notes:} Each panel plots conditional estimates of the effects of beliefs on outcomes $\E[\beta_i \mid \cdot]$. Panel A (panel experiments) conditions on the absolute value of observed belief changes $\abs{\Delta X_i}$. Panels B and C (active and passive control experiments) condition on the rank of the estimated learning rate $\alpha_i$, which measures responsiveness to experimental information. In all panels, smaller values on the horizontal axis correspond to individuals who update their beliefs less. Confidence intervals are twice the bootstrap standard error.
\end{quote}
\end{figure}
 
{
\clearpage \singlespacing \printbibliography
}

\clearpage

\appendix
\renewcommand{\thefigure}{\thesection.\arabic{figure}}
\renewcommand{\thetable}{\thesection.\arabic{table}}

\counterwithin{figure}{section}
\counterwithin{table}{section}

\addtocontents{toc}{\protect\setcounter{tocdepth}{4}}

\clearpage
{\onehalfspacing
\renewcommand{\contentsname}{Contents of Online Appendix}
\pagenumbering{roman} \tableofcontents \clearpage \pagenumbering{arabic}
}
\clearpage
\renewcommand*{\thepage}{A.\arabic{page}}

\section{Learning Rate Models} \label{sec:learning_rate_models}

This appendix provides microfoundations for the belief updating model in \eqref{eq:beliefs_po_signal}. Section \ref{sec:generalized_updating} introduces a general class of updating rules that preserve rank invariance---the minimal condition required for LLS to have equal weights. Several behavioral deviations from Bayesian updating fall within this class.

\subsection{Generalized Updating and Rank Invariance} \label{sec:generalized_updating}

For LLS to have equal weights across individuals, we need \emph{rank invariance}: the relative magnitude of belief updates must be consistent across signals. Formally, let $X_i(s, x_0)$ denote individual $i$'s posterior given signal $s$ and prior $x_0$. If $X_i(s_A, x_0) > X_j(s_A, x_0)$ for some signal $s_A$, then $X_i(s_B, x_0) > X_j(s_B, x_0)$ for any other signal $s_B$ on the same side of the prior.

Intuitively, people who respond more strongly to one signal also respond more strongly to another. Suppose Chris and Dianne have the same prior and receive the same signal above their prior. If Chris's posterior is higher than Dianne's, rank invariance requires that Chris's posterior would also be higher if both received a different signal that was also above their prior.

This section introduces a general class of updating rules that preserve rank invariance. The Bayesian baseline is a special case, but so are several behavioral deviations.

\subsubsection{General Learning Rate Updating}

Consider a general updating rule:
\begin{equation}
X_i(s, x_0) = x_0 + \alpha_i \times f(s,x_0)
\label{eq:general_updating}
\end{equation}

where $f(\cdot, X_i^0)$ is any function monotonic in the signal with $f(X_i^0, X_i^0) = 0$. The function $f$ is a link function that allows for nonlinearity in the effects of signals on the posterior belief. The individual parameter $\alpha_i > 0$ controls differences in updating between people with the same prior and signal.

The leading special cases is when $f(s,x_0)$ is the difference $s - x_0$. Then, \eqref{eq:general_updating} reduces to:
\begin{equation}
X_i(s, x_0) = \alpha_i s + (1-\alpha_i) x_0
\label{eq:linear_learning_rate}
\end{equation}

Which is the simple linear updating rule generated by Bayesian updating, among others.

\subsubsection{Bayesian Learning as a Baseline} \label{sec:bayesian_learning}

The literature often motivates the weighted-average expression in \eqref{eq:linear_learning_rate} with a Bayesian learning model featuring normally distributed beliefs \citep{ballaelliott22,cullenHowMuch22}. Consider individuals with uncertain prior beliefs. The subjective probability that the variable $X_i$ takes value $x$ is given by the density of $\mathcal{N}\left(X^0_i , \sigma^2_{iX} \right)$. We interpret $X_i^0$ as the mean of the prior distribution and call it the \emph{prior belief}.

People observe a signal $S_i$ drawn from $\N(S_i^*, \sigma^2_{iS})$. The variances reflect subjective (inverse) precision: people for whom $\sigma^2_{iS} / \sigma^2_{iX}$ is large think their prior is more precise than the signal, while those with small ${\sigma^2_{iS}}/{\sigma^2_{iX}}$ think the signal is more precise than their prior.

The posterior distribution is:
\begin{gather}
        \N\left( \left( 1-\alpha_i \right)X^0_i + \alpha_i S_i, \frac{\sigma^2_{iS}\sigma^2_{iX}}{\sigma^2_{iS} + \sigma^2_{iX}}\right) \\
        \text{where}~\alpha_i \equiv \frac{\sigma^2_{iX}}{\sigma^2_{iS} + \sigma^2_{iX}}
\end{gather}
The mean of the posterior is a weighted average of the prior $X^0_i$ and signal $S_i$, with weights determined by relative precision.\footnote{See \citet{robertBayesianChoice07} or \citet{hoffFirstCourse09} for textbook treatments.} We call this mean the \emph{posterior belief} $X_i$. The prior $X^0_i$, signal $S_i$, and posterior $X_i$ are thus related by:
\begin{equation}
    X_i =  (1-\alpha_i) X^0_i + \alpha_i S_i
\end{equation}
which generates the potential outcomes for beliefs in \eqref{eq:beliefs_PO}.

There is direct empirical support for this foundation. \citet{rothRiskExposure22} find that belief updating is driven entirely by people who report being \q{very unsure}, \q{unsure}, or \q{somewhat unsure}. Those who are \q{sure} or \q{very sure} do not update. Similarly, \citet{rothHowExpectations20} find that people less confident in their priors update roughly twice as much. \citet{kerwinNavigatingAmbiguity23} find that people with less precise priors update more in a more general model.

\subsubsection{Linear Deviations from Bayesian Updating}

Any model where updating takes the linear form \eqref{eq:linear_learning_rate} satisfies rank invariance, regardless of how $\alpha_i$ is determined. Three behavioral deviations retain this structure.

\paragraph{Diagnostic Expectations}

\citet{gretherBayesRule80} models deviations from Bayesian updating by raising the likelihood and prior to different powers. Under normal-normal learning, this rescales the effective variances of prior and signal. The learning rate $\alpha_i$ then depends on \q{behavioral} variances rather than true variances, but updating remains linear. People can vary in the heuristics they use to update and could either over-update or under-update, as long as these differences are reflected in the learning rate.

\paragraph{Rational Inattention}

\citet{fusterExpectationsEndogenous22} develop a rational inattention model where the learning rate $\alpha_i$ depends on the marginal cost of attention and the value of information. The posterior is still a weighted average of signal and prior, but the weight reflects optimal attention allocation rather than prior precision. Importantly, some people can have a \q{corner} solution and ignore the information entirely, which allows some people to have $\alpha_i = 0$.

\paragraph{Anchoring on Signal or Prior}

In anchoring models \citep{gabaixBehavioralInattention19}, people form posteriors as a weighted average of a Bayesian posterior and an anchor:
\begin{equation}
X_i(s,x_0) = \kappa_i A(s, x_0) + (1-\kappa_i)X_i^{B}(s,x_0)
\label{eq:anchoring}
\end{equation}
where $X_i^{B}(s,x_0) = \alpha_i s + (1-\alpha_i)x_0$ is the Bayesian posterior. When the anchor is itself a weighted average of the signal and prior ($A(s, x_0) = \gamma s + (1-\gamma) x_0$ for $\gamma \in [0,1]$) substitution yields:
\begin{align}
X_i(s,x_0)
&= \underbrace{[\kappa_i \gamma + (1-\kappa_i)\alpha_i]}_{\alpha_i^{\text{eff}}} s + [1-\alpha_i^{\text{eff}}]x_0
\end{align}
The anchored posterior is simply a weighted average with effective learning rate $\alpha_i^{\text{eff}}$. Two leading special cases are anchoring on the signal or anchoring on the prior. Both preserve the linear structure.

However, anchoring on a constant $A(s, x_0) = \bar{x}$ does not have this learning-rate representation and violates rank invariance.

\subsubsection{Representative Estimates When Everyone Updates} \label{sec_ap_compliers}

Many of the updating models discussed above have micro-foundations that ensure $\alpha_i \in (0,1)$. This includes standard Bayesian updating, diagnostic expectations, and anchoring on the prior or signal. When learning rates are strictly positive, everyone updates at least somewhat in response to the signal. There are no never-takers who ignore information entirely.

When everyone responds to the signal, LLS identifies an average over the full experimental sample. In contrast, if some individuals have $\alpha_i = 0$ and completely ignore the signal, both LLS and TSLS only identify an average over those who update. Individuals who never respond receive zero weight in both estimators. The LLS estimand then downgrades to a local average treatment effect (LATE): an average only among compliers who respond to the instrument.

This is closely related to a more general identification result emphasized by \citet{heckmanLocalInstrumental01} that the ATE is identified when there are values of the instrument $z,z'$ with propensity scores of zero and one (everyone is a $z \to z'$ complier).
Targeting an ATE-like parameter requires that the instrument affects the endogenous variable for all individuals. When $\alpha_i > 0$ for everyone, the signal moves everyone's beliefs (at least slightly) and so everyone is a complier.

\paragraph{Downgrading to LATE with Rational Inattention}

The rational inattention framework allows $\alpha_i = 0$ for some individuals. When the cost of attention exceeds the value of information, optimal behavior is to ignore the signal entirely. In this case, the LLS estimand downgrades from a structural APE or ATE-like parameter (representative of all experimental subjects) to a LATE parameter (representative  of those who update). The estimand remains interpretable; it simply averages over compliers rather than the full sample.

This downgrade does not change the fundamental difference between LLS and TSLS. Among those who update ($\alpha_i > 0$), LLS continues to place equal weight while TSLS weights by the size of the update. The weighting distinction persists regardless of whether the average is over everyone or only over compliers.

\subsubsection{Estimation with Nonlinear Updating}

Under linear updating, we can control for the prior linearly and condition nonparametrically only on the learning rate $\alpha_i$. This simplification extends to all linear learning rate models. Under nonlinear updating \eqref{eq:general_updating}, the estimation approach changes. The prior no longer enters linearly, so we cannot separate conditioning on $\alpha_i$ from conditioning on $x_0$. Further, the learning rate $\alpha_i$ is not directly identified from the ratio of the belief update to the difference between the signal and the prior. Instead, the conditional rank of the learning rate is identified from the sign-corrected conditional rank of the posterior given the prior and the signal:

\begin{equation}
 R_i \equiv \text{rank}\bp{\alpha_i \mid X_i^0} =  \begin{cases}
   \text{rank}\bp{Y_i \mid X_i^0, S_i} & \text{if } S_i > X_i^0 \\
   1 - \text{rank}\bp{Y_i \mid X_i^0, S_i} & \text{if } S_i < X_i^0
\end{cases}
\end{equation}

The local regression must condition on the full vector $C_i = (R_i, X_i^0, S_i(A), S_i(B))$ nonparametrically, or make alternative simplifying assumptions.
 
\section{Nonlinearity, Convex Combinations, and Discrete Slopes} \label{app:nonlinearity_details}

This appendix provides additional detail on the interpretation of the LLS estimand under nonlinear outcome functions. The main text (Section \ref{sec:comparing_estimators}) establishes that LLS delivers a representative average regardless of functional form, while linearity adds a structural interpretation permitting extrapolation. This appendix elaborates on three related points.

\subsection{Convex Combinations and Magnitudes} \label{app:convex_combinations}

When simple averages are not identified, researchers sometimes target parameters from a broader class of convex combinations (weighted averages with non-negative weights that sum to one). Under a weaker \q{signal monotonicity} assumption, the unweighted average is not generally identified by LLS, but TSLS can still identify a convex combination of individual effects \citep{vz_aeri}.

\citet{mogstadInstrumentalVariables24} note that convex combinations are informative only about the sign of effects, and only when every individual effect has the same sign. Parameters that can be an arbitrary convex combination are generally uninformative about magnitudes. An equally weighted average, by contrast, is informative about both sign and magnitude. This is a key advantage of targeting the LLS estimand under learning-rate updating: it delivers not just a convex combination, but a specific, interpretable average.

\subsection{Discrete Slopes Versus Derivatives} \label{app:slopes_vs_derivatives}

With only two potential beliefs per person, we observe only the average slope $\beta_i$ over each person's belief interval $[X_i(B), X_i(A)]$, not the full shape of $G_i(\cdot)$.\footnote{In general, binary instruments do not identify the nonlinearity of individual response functions without further assumptions \citep{brinchLATEDiscrete17}.}

The difference between the average slope $\E[\beta_i]$ and the non-parametric ATE $\E[G'_i(X_i)]$ is the difference between slopes over discrete changes and derivatives. The ATE averages derivatives; $\E[\beta_i]$ averages slopes over discrete changes. In this sense, $\E[\beta_i]$ can be interpreted as a discrete approximation to the ATE.\footnote{Researchers may also be interested in an alternative parameter $\E[G'_i(X_i^0)]$, where derivatives are evaluated at the prior rather than the posterior. Under linearity, $G'_i(x) = \tau_i$ for all $x$, so $\E[\beta_i]$, the ATE, and this alternative all coincide. Under nonlinearity, the gap between the discrete-slope parameter $\E[\beta_i]$ and derivative-based parameters depends on both the curvature of $G_i(\cdot)$ and the width of the belief interval $X_i(A) - X_i(B)$.}

As the potential beliefs get closer together, the discrete slope converges to the derivative: $\beta_i \to G'_i(X_i)$ when $X_i(A) - X_i(B) \to 0$. This would require picking signals very close to each other (in the active case) or very close to the prior (in the passive case).
This is likely unattractive in practice, since using a pair of signals that are close to each other would lead to a weak first stage.

\subsubsection{Identifying Derivatives Without Linearity} \label{app_np_identifition}

Richer experimental variation could identify derivative-based parameters without assuming linearity. Experiments with $K$ signal values can identify degree-$(K-1)$ polynomial approximations to the response function \citep{mastenIdentificationInstrumental16}. Experiments with continuously distributed signals can identify the average structural function $\E[G_i(x)]$ using nonparametric control function methods \citep{heckmanChapter7107a, imbensIdentificationEstimation09}.

These approaches require substantially richer data than the two-arm experiments typical in current practice. The practical implication is that researchers using standard experimental designs face a choice: either maintain linearity and interpret $\E[\beta_i]$ as the ATE/APE, or relax linearity and interpret $\E[\beta_i]$ as a representative average of slopes over the discrete belief changes induced by the experiment.

\subsection{Nonlinearity Affects Interpretation, Not the Case for Equal Weights}

The key point from Section \ref{sec:comparing_estimators} bears repeating: nonlinearity affects the interpretation of individual slopes $\beta_i$, but it does not affect the difference between LLS and TSLS. Both estimators aggregate the same individual slopes; they differ only in how they weight those slopes. LLS uses equal weights; TSLS uses weights proportional to belief updating.

Whether or not outcomes are linear in beliefs, researchers who want a representative summary of heterogeneous effects should prefer equal weights. Linearity is an additional assumption that strengthens the interpretation of the equally-weighted average; it is not a precondition for preferring equal weights over unequal weights.

\section{Proofs and Derivations}

This section contains proofs and derivations.

\subsection{Derivations of Weights} \label{sec:weight_derivations}

This section provides derivations for the weights reported in Section \ref{sec:lit_weights}. All three derivations use the definition of individual slopes from equation \eqref{eq:defn_indiv_slopes}:
\begin{equation}
	\beta_i \equiv \frac{G_i(X_i(A)) - G_i(X_i(B))}{X_i(A) - X_i(B)} \tag{\ref{eq:defn_indiv_slopes}}
\end{equation}
This definition implies that $G_i(X_i(A)) - G_i(X_i(B)) = \beta_i (X_i(A) - X_i(B))$ for any outcome function $G_i(\cdot)$.

\subsubsection{Weights in the Panel Specification}
\begin{assumption}\label{assumption:panel} Panel Assumptions.
\begin{enumerate}
\item Panel Outcomes: The panel outcome equation \eqref{eq:panel_outcome} holds.
\begin{equation}
 Y_{it} = G_i(X_{it}) + \gamma_t  \tag{\ref{eq:panel_outcome}}
\end{equation}
\item Relevance: There is variation in beliefs over time $\var \bs{\Delta X_i} > 0 $.
\item Existence: The relevant moments exist and are finite.
\end{enumerate}
\end{assumption}

The parsimonious specification in the panel data model in \eqref{eq:panel_regression} is given by:
\begin{align}
    \beta^{Panel} &= \frac{\cov \bs{\Delta Y_i,  \Delta X_i}}{\var \bs{\Delta X_i}} \\
    \intertext{Substitute the outcome equation \eqref{eq:panel_outcome}:}
    &= \frac{\cov \bs{G_i(X_{i1}) - G_i(X_{i0}) + \gamma_1 - \gamma_0,  \Delta X_i}}{\var \bs{\Delta X_i}}
    \intertext{Apply the definition of individual slopes. Since $\Delta X_i = X_{i1} - X_{i0}$, we have $G_i(X_{i1}) - G_i(X_{i0}) = \beta_i \Delta X_i$:}
    &= \frac{\cov \bs{\beta_i \Delta X_i + \gamma_1 - \gamma_0,  \Delta X_i}}{\var \bs{\Delta X_i}}
    \intertext{From definitions of covariance and variance; $\cov(a,b) = \E\bs{a(b-\E(b))}$:}
    &=\frac{\E \bs{\beta_i \Delta X_i  \bp{\Delta X_i - \E \bs{\Delta X_i} } }}{\E \bp{\Delta X_i  \bp{\Delta X_i - \E \bs{\Delta X_i} }}} \\
    \intertext{To express this as a weighted average of individual effects, rearrange:}
    &=\E \bs{\beta_i \cdot  \frac{ \Delta X_i  \bp{\Delta X_i - \E \bs{\Delta X_i} } }{\E \bp{\Delta X_i  \bp{\Delta X_i - \E \bs{\Delta X_i} }}}}
\end{align}

This gives the weights $\omega_i(\text{Panel}) \propto \Delta X_i \bp{\Delta X_i - \E \bs{\Delta X_i} } $, which are normalized to integrate to one.

\subsubsection{Weights in the Active Control Specification}

\begin{assumption}\label{assumption:active} Active Control Assumptions.
\begin{enumerate}
\item Nonparametric Outcomes: The outcome model in equation \eqref{eq:outcome} holds.
\begin{align}
    Y_i    &= G_i(X_i) \tag{\ref{eq:outcome}}
\end{align}
\item Learning rate updating: The belief potential outcomes in equation \eqref{eq:beliefs_PO} hold.
\begin{align}
X_i(z) & = \alpha_i \bp{S_i(z) - X_i^0} + X^0_i
 \tag{\ref{eq:beliefs_PO}}
\end{align}
\item Relevance: There is variation in potential beliefs $\E[X_i(A) - X_i(B)] \neq 0 $.
\item Random Assignment: The treatment $Z_i$ is randomly assigned.
\item Existence: The relevant moments exist and are finite.
\end{enumerate}
\end{assumption}

Starting with the TSLS coefficient in the active control design:
\begin{align}
    \beta^{\text{TSLS}} &= \frac{\E[Y_i \mid Z_i = A] - \E[Y_i \mid Z_i = B]}{\E[X_i \mid Z_i = A] - \E[X_i \mid Z_i = B]} \\
\intertext{From the outcome equation \eqref{eq:outcome} and random assignment:}
    &= \frac{\E[G_i(X_i(A))] - \E[G_i(X_i(B))]}{\E[X_i(A)] - \E[X_i(B)]} \\
    \intertext{Apply the definition of individual slopes: $G_i(X_i(A)) - G_i(X_i(B)) = \beta_i(X_i(A) - X_i(B))$:}
    &= \frac{\E[\beta_i (X_i(A) - X_i(B))]}{\E[X_i(A) - X_i(B)]} \\
    \intertext{To express this as a weighted average of individual effects, rearrange:}
    &= \E\left[\beta_i \cdot \frac{X_i(A) - X_i(B)}{\E[X_i(A) - X_i(B)]}\right]
\end{align}

This gives us the weights $\omega_i(\text{Active}) \propto X_i(A) - X_i(B)$, which are normalized to integrate to one.

\subsubsection{Weights in the Passive Control Specification} \label{sec:derivation_passive_tsls}

\begin{assumption}\label{assumption:passive} Passive Control Assumptions.
\begin{enumerate}

\item Nonparametric Outcomes: The outcome model in equation \eqref{eq:outcome} holds.
\begin{align}
    Y_i    &= G_i(X_i) \tag{\ref{eq:outcome}}
\end{align}
\item Learning rate updating: The belief potential outcomes in equation \eqref{eq:beliefs_PO} hold.
\begin{align}
X_i(z) & = \alpha_i \bp{S_i(z) - X_i^0} + X^0_i
 \tag{\ref{eq:beliefs_PO}}
\end{align}
\item Relevance: There is variation in potential beliefs $\E[X_i(A) - X_i(B)] \neq 0 $.
\item Random Assignment: The treatment $Z_i$ is randomly assigned.
\item Existence: The relevant moments exist and are finite.
\item Passive control: Treatment arm $B$ does not receive any signal: $S_i(B) \equiv X^0_i$.
\end{enumerate}
\end{assumption}

In the passive control design, the exposure-weighted instrument is defined as:
\begin{align}
    T^{ex}_i &\equiv T_i(S_i(A) - X_i^0) \tag{\ref{eq:exposure_reg}} \\
    \intertext{Since we are interested in coefficients on $T^{ex}_i$ in regressions that control for $S_i(A) -  X_i^0$, we can appeal to FWL and instead consider the coefficients on the residualized $\wt{T}^{ex}_i$. To construct this residual, regress $T^{ex}_i$ on $(S_i(A) -  X_i^0)$ and a constant:}
    \theta &= \frac{\cov(T^{ex}_i, S_i(A) -  X_i^0)}{\var(S_i(A) -  X_i^0)} \\
    &= \frac{\E[T_i(S_i(A) -  X_i^0)^2] - \E[T_i]\E[(S_i(A) -  X_i^0)^2]}{\var(S_i(A) -  X_i^0)} \\
    \intertext{Since $T_i$ is binary and independent of $(S_i(A) -  X_i^0)$ by random assignment:}
    \theta &= \frac{\E[T_i]\var(S_i(A) -  X_i^0)}{\var(S_i(A) -  X_i^0)} = \E[T_i] \\
    \intertext{The recentered instrument is then the residual from this regression:}
    \wt{T}^{ex}_i &= T^{ex}_i - \theta(S_i(A) -  X_i^0) \\
    &= (T_i - \E[T_i])(S_i(A) -  X_i^0)
\end{align}

Since $ \E \bs{\wt{T}^{ex}_i} = 0$, the TSLS coefficient is:
\begin{align}
    \beta^{\text{Passive}} &= \frac{\E[\wt{T}^{ex}_i Y_i]}{\E[\wt{T}^{ex}_i X_i]}
\end{align}

The denominator is:
\begin{align}
    \E[\wt{T}^{ex}_i X_i] &= \E[(T_i - \E[T_i])(S_i(A) -  X_i^0) \cdot X_i] \\
    \intertext{Plugging in the potential beliefs for $X_i$:}
    &= \E[T_i](1-\E[T_i])\E[(S_i(A) -  X_i^0)(X_i(A) - X^0_i)] \\
    \intertext{Using the definition of $X_i(A)$ from \eqref{eq:beliefs_PO} to simplify further yields:}
    &= \E[T_i](1-\E[T_i])\E[\alpha_i(S_i(A) - X^0_i)^2]
\end{align}

For the numerator, apply the outcome equation and random assignment:
\begin{align}
    \E[\wt{T}^{ex}_i Y_i] &= \E[(T_i - \E[T_i])(S_i(A) -  X_i^0) \cdot G_i(X_i)] \\
    &= \E[T_i](1-\E[T_i])\E[(S_i(A) - X^0_i)(G_i(X_i(A)) - G_i(X^0_i))]
    \intertext{Apply the definition of individual slopes. In passive designs $X_i(B) = X^0_i$, so $G_i(X_i(A)) - G_i(X^0_i) = \beta_i(X_i(A) - X^0_i) = \beta_i \alpha_i (S_i(A) - X^0_i)$:}
    &= \E[T_i](1-\E[T_i])\E[\beta_i\alpha_i(S_i(A) - X^0_i)^2]
\end{align}

Thus, the TSLS coefficient is:
\begin{align}
    \beta^{\text{Passive}} &= \frac{\E[T_i](1-\E[T_i])\E[\beta_i\alpha_i(S_i(A) - X^0_i)^2]}{\E[T_i](1-\E[T_i])\E[\alpha_i(S_i(A) - X^0_i)^2]} \\
    &= \E\left[\beta_i \cdot \frac{\alpha_i(S_i(A) - X^0_i)^2}{\E[\alpha_i(S_i(A) - X^0_i)^2]}\right]
\end{align}

This gives us the weights $\omega_i(\text{Passive}) \propto \alpha_i(S_i(A) - X^0_i)^2$, which are normalized to integrate to one.

\subsection{Main Identification Results} \label{sec:identification_derivation}

The key identification insight across all three designs is that by appropriately conditioning on observables, we can isolate variation in beliefs that is driven solely by exogenous treatment assignment. This creates local comparisons where beliefs effectively take only two values, making each regression equivalent to a simple difference in conditional means. This section proves that these local regressions recover average partial effects.

\begin{proposition}[Binary Regression Property] \label{prop:binary_regression}
Consider a linear regression of $Y$ on $X$ where $X$ takes only two values, $x_1$ and $x_2$. Then the regression coefficient $\beta$ equals:
\begin{equation}
\beta = \frac{\E[Y \mid X=x_2] - \E[Y\mid X=x_1]}{x_2 - x_1}
\end{equation}
\end{proposition}

\begin{proof}
The regression coefficient is defined as:
\begin{align}
\beta = \frac{\Cov(Y,X)}{\Var(X)}
\end{align}

Let $p = \Pr[X = x_2]$. Then:
\begin{align}
\Var(X) &= \E[(X - \E[X])^2] \\
&= p(1-p)(x_2 - x_1)^2
\end{align}

For the covariance:
\begin{align}
\Cov(Y,X) &= \E[(Y - \E[Y])(X - \E[X])] \\
&= p(1-p)(x_2 - x_1)(\E[Y\mid X=x_2] - \E[Y\mid X=x_1])
\end{align}

Therefore:
\begin{align}
\beta = \frac{\Cov(Y,X)}{\Var(X)} = \frac{\E[Y\mid X=x_2] - \E[Y\mid X=x_1]}{x_2 - x_1}
\end{align}
\end{proof}

\subsubsection{Identification in Panel Experiments} \label{sec:panel_id_proofs}

\renewcommand{\theassumption}{1A}
\begin{assumption} \label{assumption:panel_additional}
Maintain the panel assumptions \ref{assumption:panel}. Additionally
\begin{enumerate}[label=\roman*.]
\item Either $\mathbb{P}[\Delta X_i = 0] > 0$ (control group exists), or $\Delta X_i$ has positive density in a neighborhood of zero (as in \cite{grahamIdentificationEstimation12}).
\item Nonlinear outcome: Relax the outcome equation to
    \begin{align}
    Y_{it}(x) &= G_i(x) + \gamma_t \tag{\ref{eq:panel_outcome}}
    \end{align}
\end{enumerate}
\end{assumption}
\renewcommand{\theassumption}{\arabic{assumption}} 

\begin{proposition}[Panel Identification]
Under Assumption \ref{assumption:panel_additional}, for any $x \neq 0$ in the support of $\Delta X_i$:

\begin{equation}
 \frac{\E \bs[\big]{G_i(X_i^0 + x) - G_i(X_i^0) \mid \Delta X_i = x }}{x}  =
 \frac{\E \bs[\big]{\Delta Y_i \mid  \Delta X_i = x} - \E \bs[\big]{\Delta Y_i \mid \Delta X_i = 0}}{x}
\end{equation}

If $G_i(X_{it}) = \tau_i X_{it} + U_i$ as in \eqref{eq:panel_outcome}, the estimand simplifies further to $\E[\tau_i \mid  \Delta X_i = x]$.

\end{proposition}

\begin{proof}
By Proposition \ref{prop:binary_regression}, the regression of $\Delta Y_i$ on $\Delta X_i$ conditional on $\Delta X_i \in \{0, x\}$ has coefficient:
\begin{align}
\beta(x) = \frac{\E[\Delta Y_i \mid  \Delta X_i = x] - \E[\Delta Y_i \mid  \Delta X_i = 0]}{x}
\end{align}

For individuals with $\Delta X_i = x$, we have $X_{i1} = X_{i0} + x$. Thus:
\begin{align}
\E[\Delta Y_i \mid  \Delta X_i = x] &= \E[G_i(X_{i0} + x) - G_i(X_{i0}) \mid  \Delta X_i = x] + \Delta \gamma
\end{align}

For those with $\Delta X_i = 0$, we have $X_{i1} = X_{i0}$, giving:
\begin{align}
\E[\Delta Y_i \mid  \Delta X_i = 0] &= \E[G_i(X_{i0}) - G_i(X_{i0}) + \Delta \gamma \mid  \Delta X_i = 0] \\
&= \Delta \gamma
\end{align}

Taking the difference:
\begin{align}
\E[\Delta Y_i \mid  \Delta X_i = x] - \E[\Delta Y_i \mid  \Delta X_i = 0] &= \E[G_i(X_{i0} + x) - G_i(X_{i0}) \mid  \Delta X_i = x]
\end{align}

Dividing by $x$ completes the proof:
\begin{align}
\frac{\E[\Delta Y_i \mid  \Delta X_i = x] - \E[\Delta Y_i \mid  \Delta X_i = 0]}{x} &= \frac{\E[G_i(X_{i0} + x) - G_i(X_{i0}) \mid  \Delta X_i = x]}{x}
\end{align}

\end{proof}

The necessity of a control group (\ref{assumption:panel_additional}) is not unique to the LLS estimator, but is instead a necessary condition for the data to be informative about the $\tau_i$. Formally:

\begin{proposition}[Necessity]
If Assumption \ref{assumption:panel_additional}.i fails, the identified sets for $\gamma_t$ and each $\tau_i$ are the real line.
\label{prop:panel_control_necessity}
\end{proposition}

\begin{proof}
Suppose Assumption \ref{assumption:panel_additional}.i fails, such that $\Delta X_i$ is bounded away from zero. Then for any candidate intercept $a$, define:
\begin{equation}
    B_i(a) \equiv \frac{\Delta Y_i - a}{\Delta X_i}
\end{equation}
The pair $(a, B_i(a))$ is observationally equivalent to $(\gamma_1 - \gamma_0, \tau_i)$ since they generate the same joint distribution of $(\Delta Y_i, \Delta X_i)$ and satisfy $\E[\Delta Y_i - a - B_i(a) \Delta X_i \mid \Delta X_i] = 0$.
We can repeat the exercise by first choosing any $i'$ and any $B_{i'}$. Chose $a(B_{i'}) \equiv \frac{ \Delta Y_i}{ B_{i'} \Delta X_i}$ and then chose the remaining $B_i$ as above.

Thus the identified sets for $\gamma_1 - \gamma_0$ and $\tau_i$ are the real line. Chose an arbitrary $\gamma_1 - \gamma_0$ or an arbitrary $\tau_{i'}$ for some $i'$ and there are values for the remaining parameters that rationalize the data.
\end{proof}

The \q{control group} is crucial to identify $\gamma_t$ in this flexible model. If there is no control group it is necessary to consider adding additional assumptions. One solution would be simply to assume that $\gamma_t = 0$ such that causal effects can be directly identified from with-individual first-differences.

\subsubsection{Identification in Active Experiments} \label{sec:proof_active_ape_ID}

\renewcommand{\theassumption}{2A}
\begin{assumption}\label{assumption:active_ID}
The active control design maintains assumptions \ref{assumption:active} from above, with the following modificiations:
\begin{enumerate}[label=\roman*.]
\item Relevance: $S_i(A) \neq S_i(B)$ and $\alpha_i > 0$.
\item Nonlinear outcome: Use the general form of potential outcomes
]
    \[ Y_{i}(x) = G_i(x) \tag{\ref{eq:outcome}} \]
    \end{enumerate}
\end{assumption}
\renewcommand{\theassumption}{\arabic{assumption}}

\begin{proposition}[Active Control Identification] \label{prop:active_control_ID}
Under Assumption \ref{assumption:active}, for any value $c$ of the control vector $C_i \equiv \bs{\alpha_i \; X_i^0 \; S_i(A)\;  S_i(B)}$:

\begin{equation}
\frac{\E \bs{G_i(x_A) - G_i(x_B) \mid C_i = c}}{x_A - x_B} = \frac{\cov \bs{Y_i, X_i \mid C_i = c}}{\var \bs{X_i \mid C_i = c}}
\end{equation}
where $x_A$ and $x_B$ are the deterministic belief values for individuals with $C_i = c$. In the special case where $G_i(X_{it}) = \tau_i X_{it} + U_i$ as in \eqref{eq:outcome}, the estimand simplifies further to $\E[\tau_i \mid C_i = c]$.

\end{proposition}

\begin{proof}

Since $C_i$ includes $\alpha_i$, $X_i^0$, $S_i(A)$, and $S_i(B)$, the potential beliefs take the same value for all individuals with $C_i = c$.
\begin{align}
X_i(A) &= X_i^0 + \alpha_i(S_i(A) - X_i^0) \\
X_i(B) &= X_i^0 + \alpha_i(S_i(B) - X_i^0)
\end{align}

Thus, conditional on $C_i = c$, the observed belief $X_i$ equals either $X_i(A) = x_A$ or $X_i(B) = x_B$ depending solely on the randomly assigned treatment $Z_i$. By Proposition \ref{prop:binary_regression}, the regression of $Y_i$ on $X_i$ conditional on $C_i = c$ has coefficient:
\begin{align}
\beta(c) = \frac{\E[Y_i \mid  X_i = x_A, C_i = c] - \E[Y_i \mid  X_i = x_B, C_i = c]}{x_A - x_B}
\end{align}
Relevance guarantees that $x_A \neq x_B$ and therefore $X_i = x_A$ if and only if $Z_i = A$, and $X_i = x_B$ if and only if $Z_i = B$. This yields
\begin{align}
\beta(c) = \frac{\E[Y_i \mid  Z_i = A, C_i = c] - \E[Y_i \mid  Z_i = B, C_i = c]}{x_A - x_B}
\end{align}
Then, since $Z_i$ is randomly assigned, we have:
\begin{align}
\E[Y_i \mid  Z_i = A, C_i = c] - \E[Y_i \mid  Z_i = B, C_i = c] &= \E[G_i(x_A) - G_i(x_B) \mid  C_i = c]
\end{align}
Dividing by ${x_A - x_B}$ completes the proof:
\begin{align}
 \frac{\cov \bs{Y_i, X_i \mid C_i = c}}{\var \bs{X_i \mid C_i = c}} = \frac{\E \bs{G_i(x_A) - G_i(x_B) \mid C_i = c}}{x_A - x_B}
\end{align}
\end{proof}

\subsubsection{Identification in Passive Experiments} \label{app_ID_passive_APE}

Once the control vector $C_i$ is available, the proof in the passive case is identical to the active case. By convention, we set $S_i(B) = X_i^0$ in the passive case, so $S_i(B)$ can be omitted from the control vector $C_i$. The difference lies in constructing the first element of the control vector $C_i$. The identification challenge in the passive case is that the learning rate $\alpha_i$ is unknown for the control group that does not receive information. There are two possible approaches in this case

\begin{assumption}\label{assumption:passive_common_signal} Common signal variance and observed prior variance.

\begin{enumerate}
\item Let $\alpha_i = \frac{{\sigma^2_X}_i}{{\sigma^2_X}_i  + \sigma^2_S}$ with $\sigma^2_S$ common across individuals.
\item The researcher knows ${\sigma^2_X}_i$.
\end{enumerate}

\end{assumption}

In normal-normal Bayesian updating, $\alpha_i = \frac{{\sigma^2_X}_i}{{\sigma^2_X}_i  + \sigma^2_S}$, where ${\sigma^2_X}_i$ is the variance of the prior belief $X_i^0$ and $\sigma^2_S$ is the variance of the signal $S_i$. The first assumption, that $\sigma^2_S$ is common, means that people all think the signal is equally informative. The second assumption is about the design of the experiment and simply states that the variance of the prior distribution is elicited as in \citet{kumarEffectMacroeconomic23}.

\begin{assumption}\label{assumption:passive_selectiononobservables} Belief updates can be predicted from observables (i.e. no unobservable heterogeneity in updating).

\begin{enumerate}
\item There is some function $f$ with (estimable) parameters $\theta$ such that $X_i(A) = f(\theta, W_i)$
\end{enumerate}

For example, if $f$ is a linear function of $W_i$ as in \citet{ballaelliott22,cantoniProtestsStrategic19}, then $X_i(A) = W'_i \theta $. Since $Z_i$ is randomly assigned,  $\theta$ is identifed from a regression on the sample assigned to $A$.

\end{assumption}

If assumption \ref{assumption:passive_common_signal} does not hold, researchers who would like to estimate the APE must make a strong assumption that there are sufficiently rich covariates to predict all of the heterogeneity in belief updating. This is in contrast with the active control designs, that use the observed updates as a \q{revealed preference} measure of peoples' learning rates.

\renewcommand{\theassumption}{3A}
\begin{assumption}\label{assumption:passive_identification}
The passive control design maintains assumptions \ref{assumption:passive} from above, with the following modificiations:
\begin{enumerate}[label=\roman*.]
\item Relevance: $S_i(A) \neq X_i^0$ and $\alpha_i > 0$.
\item Nonlinear outcome: Use the general form of potential outcomes
]
    \[ Y_{i}(x) = G_i(x) \tag{\ref{eq:outcome}} \]
    \item Inferred Learning Rate: Either assumption \ref{assumption:passive_common_signal} or \ref{assumption:passive_selectiononobservables} holds
\end{enumerate}
\end{assumption}
\renewcommand{\theassumption}{\arabic{assumption}} 

Assumption \ref{assumption:passive_identification} for the passive case contains Assumption \ref{assumption:active_ID} for the active case, and adds \ref{assumption:passive_identification}.iii since the learning rate is not directly identified for the control group.

\begin{proposition}[Passive Control Identification]

Under Assumption \ref{assumption:passive_identification}, for any value $c$ of the control vector $C_i$ implied by either \ref{assumption:passive_common_signal} or \ref{assumption:passive_selectiononobservables}

\begin{equation}
\frac{\E \bs{G_i(x_A) - G_i(x_B) \mid C_i = c}}{x_A - x_B} \equiv \frac{\cov \bs{Y_i, X_i \mid C_i = c}}{\var \bs{X_i \mid C_i = c}}
\end{equation}

\end{proposition}

\begin{proof}

Under Assumption \ref{assumption:passive_common_signal}, $\alpha_i$ is a one-to-one function of ${\sigma^2_X}_i$. Thus conditioning on ${\sigma^2_X}_i$ or its rank is equivalent to conditioning on $\alpha_i$ and so conditional on $C_i \equiv \bs{\text{rank}\bp{ {\sigma^2_X}_i} \; X_i^0 \; S_i(A)}$, $X_i(A)$ and $X_i(B)$ are deterministic. The rest of the proof is identical to the active case.

Under Assumption  \ref{assumption:passive_selectiononobservables}, $X_i(A)$ in the control group is known from $f(\theta, W_i)$. To maintain similar arguments as the other cases, notice then that this implies that $\alpha_i$ is identified from $\frac{f(\theta, W_i) - X_i^0}{S_i(A) - X^0_i}$ for the control group and directly from $\frac{X_i - X_i^0}{S_i(A) - X^0_i}$ for the treated group. Then, conditional on $C_i \equiv \bs{\alpha_i \; X_i^0 \; S_i(A)}$, $X_i(A)$ and $X_i(B)$ are deterministic. The rest of the proof is identical to the active case.

\end{proof}

In each case, integrating over the distribution of the conditioning variables recovers an average partial effect $\E \bs{\frac{ {G_i(X_i(A)) - G_i(X_i(B))}}{X_i(A) - X_i(B)}}$. In the linear case, we recover the average coefficient $\E[\tau_i]$.

\subsection{Linear Controls in a Reweighted Regression} \label{app_linear_controls_LLS}

This section shows that a reweighted linear regression that controls for $\alpha_i$ nonparametrically but only controls linearly for $X_i^0, S_i(A), S_i(B)$ also identifies the APE under the maintained assumptions.

\begin{proposition}[Linear Controls with Reweighting] \label{prop:linear_controls}
Consider the active control design with nonlinear potential outcomes $Y_i = G_i(X_i)$. Let $W_i = [X_i^0 \; S_i(A) \; S_i(B)]'$. Under Assumption \ref{assumption:active_ID}, conditional on $\alpha_i$, the weighted regression of $Y_i$ on $X_i$ and $W_i$ with weights proportional to $(S_i(A) - S_i(B))^{-2}$ yields a coefficient on $X_i$ that identifies:
\begin{equation}
\E\left[\frac{G_i(X_i(A)) - G_i(X_i(B))}{X_i(A) - X_i(B)} \bigg\mid  \alpha_i\right]
\end{equation}

In the special case where $G_i(x) = \tau_i x + U_i$, this estimand simplifies further to $\E \bs{\tau_i \mid \alpha_i}$.

The analogous result holds for the passive design under Assumption \ref{assumption:passive_identification}, with $S_i(B) = X_i^0$ by convention. The reweighted regression then has weights proportional to $(S_i(A) - X_i^0)^{-2}$.

\end{proposition}

\begin{proof}
Consider the active design; the passive case follows analogously with $S_i(B) = X_i^0$.
Appealing to FWL, consider the coefficient on $\tilde{X}_i$, the residual from the projection of $X_i$ onto $W_i = [X_i^0 \; S_i(A) \; S_i(B)]'$ conditional on $\alpha_i$. That is:
\begin{equation}
\tilde{X}_i = X_i - \L_{\alpha_i}[X_i \mid W_i]  = X_i - \E[X_i \mid W_i, \alpha_i]
\end{equation}
The second equality uses the fact that, under learning rate updating \eqref{eq:beliefs_PO}, the true conditional expectation is linear in $W_i$ conditional on $\alpha_i$:
\begin{align}
\E[X_i \mid W_i, \alpha_i] = (1-\alpha_i)X_i^0 +  \alpha_i S_i(B) + \E \bs{T_i} \alpha_i \bp{S_i(A) - S_i(B)}
\end{align}
Thus the residual is with respect to the true conditional expectation and not only the linear projection. The notation $\L_{\alpha_i}[X_i \mid W_i]$ is meant to highlight the fact that linear projection is onto $W_i$ after conditioning on $\alpha_i$.
Writing $X_i$ in a similar form shows that
\begin{align}
X_i &= (1-\alpha_i)X_i^0 +  \alpha_i S_i(B) + T_i \alpha_i \bp{S_i(A) - S_i(B)} \\
\tilde{X}_i \equiv X_i - \E[X_i \mid W_i, \alpha_i] &= \alpha_i \bp{T_i - \E \bs{T_i}}  \bp{S_i(A) - S_i(B)}
\end{align}
The weighted coefficient from regressing $Y_i$ on $\tilde{X}_i$ with weights $(S_i(A) - S_i(B))^{-2}$ is thus:
\begin{align}
\beta_{\alpha} &= \frac{\E[Y_i\tilde{X}_i (S_i(A) - S_i(B))^{-2} \mid \alpha_i]}{\E[\tilde{X}_i^2 (S_i(A) - S_i(B))^{-2} \mid \alpha_i]} \\
&= \frac{\E[Y_i \cdot \alpha_i(T_i - \E[T_i])(S_i(A) - S_i(B))^{-1} \mid \alpha_i]}{\E[\alpha_i^2(T_i - \E[T_i])^2 \mid \alpha_i]} \\
&= \frac{\E[Y_i \cdot (T_i - \E[T_i])(S_i(A) - S_i(B))^{-1} \mid \alpha_i]}{\alpha_i \E[(T_i - \E[T_i])^2 \mid \alpha_i]}
\end{align}
Now, we compute the numerator:
\begin{align}
\E \bs{Y_i \cdot \frac{(T_i - \E[T_i])}{(S_i(A) - S_i(B))} \mid \alpha_i} &= \E \bs{G_i(X_i) \cdot \frac{(T_i - \E[T_i])}{(S_i(A) - S_i(B))} \mid \alpha_i} \\
&= \E[T_i](1-\E[T_i]) \cdot \E \bs{\frac{G_i(X_i(A)) - G_i(X_i(B))}{S_i(A) - S_i(B)} \mid \alpha_i}
\end{align}
Note that the denominator simplifies to $\alpha_i \E[(T_i - \E[T_i])^2 \mid \alpha_i] = \alpha_i \E[T_i](1-\E[T_i])$ since $T_i$ is Bernoulli. Substituting both into the expression for $\beta_{\alpha}$:
\begin{align}
\beta_{\alpha} &= \frac{\E[T_i](1-\E[T_i])}{\alpha_i \E[T_i](1-\E[T_i])} \E \bs{\frac{G_i(X_i(A)) - G_i(X_i(B))}{\bp{S_i(A) - S_i(B)}} \mid \alpha_i} \\
&= \E \bs{\frac{G_i(X_i(A)) - G_i(X_i(B))}{\alpha_i\bp{S_i(A) - S_i(B)}} \mid \alpha_i}
\end{align}
Given that $X_i(A) - X_i(B) = \alpha_i(S_i(A) - S_i(B))$, the denominator simplifies further to:
\begin{align}
\beta_{\alpha} &= \E \bs{\frac{G_i(X_i(A)) - G_i(X_i(B))}{X_i(A) - X_i(B)} \mid \alpha_i}
\end{align}
This completes the proof. The derivation for the passive case is analogous, with $S_i(B) = X_i^0$ by convention. The weights are then proportional to $(S_i(A) - X_i^0)^{-2}$.
\end{proof}

\section{Estimation Details} \label{sec:estimation}
This section provides estimation details, including implementation protocols for each experimental design with specific guidance on the specification of the \q{local} regression, trimming, and bandwidth selection.

\subsection{Linear Belief Updating Simplifies Estimation}
\label{sec:linear_controls_lls}

In the replications in this paper and in many empirical settings, the sample size is small enough that it is quite demanding to non-parametrically control for the learning rate, the prior, and potential signals. Taking full advantage of the linearity in the belief updating process \eqref{eq:beliefs_PO}, it is sufficient to condition only on the learning rate and control for the prior linearly. In passive designs, or designs with person-specific high and low signals (i.e. \citet{rothRiskExposure22}), it is also necessary to reweight by the inverse of the exposure.

The specific specifications used for estimation are as follows:

\subsubsection{Local Regressions in Panel Experiments}

    Conditional on the rank of the observed change in beliefs $\Delta X_i$, regress the change in the outcome $\Delta Y_i$ on the change in beliefs $\Delta X_i$ and a constant. This is exactly the local regression in Section \ref{main_panel_APE_id}.

\subsubsection{Local Regressions in Active and Passive Control Experiments}

    In active designs, the learning rate $\alpha_i = \bp{X_i - X_i^0}/\bp{S_i - X_i^0}$ is directly observed for all individuals, since beliefs are elicited in both treatment arms. In passive designs, however, the learning rate is unobserved for the control group that receives no information. In this case, the learning rate must be imputed using one of the approaches described in Section \ref{sec:lls_passive}: either from observed prior variance under common signal precision, or from predicted beliefs using rich observables. Appendix \ref{app_ID_passive_APE} formally states the assumptions required in each case.

    Given an observed or imputed learning rate, the regression procedure is the same in active and passive control experiments.  
    Conditional on the rank of the observed learning rate $\alpha_i$, regress the outcome $Y_i$ on the posterior belief $X_i$, the prior $X_i^0$ and a constant. In the active case, if there is variation in the individual signals $S_i(A), S_i(B)$, weight the regression by $\bp{S_i(A) - S_i(B)}^{-2}$. In the passive case, weight the regression by $\bp{S_i(A) - X^0_i}^{-2}$.

\subsection{Trimming}
The estimator will perform poorly as the change in beliefs approaches zero. Trimming \q{away from zero} as in \citet{grahamIdentificationEstimation12} thus can greatly improve the performance of the estimator in finite samples.\footnote{As in \citet{grahamIdentificationEstimation12}, we can impose some mild regularity conditions (i.e. smoothness and continuity) on the function $\tau(c) = \E \bs{\tau_i \mid C_i = c }$ such that trimming does not affect the consistency of the estimators when the trimming thresholds are asymptotically zero.}

\subsubsection{Trimming in Panel Experiments}
     Chose a threshold $h^*$ and exclude observations with changes in beliefs $\abs{\Delta X_i} < h^*$. This is a special case of \citet{grahamIdentificationEstimation12}.

\subsubsection{Trimming in Active Control Experiments}
    Choose a threshold learning rate $\alpha^*$ and exclude observations with a learning rate $\alpha < \alpha^*$. If there is variation in the individual signals $S_i(A), S_i(B)$, it is also important to chose a threshold $s^*$ and exclude observations with $\bp{S_i(A) - S_i(B)}^2 < s^*$ to ensure that the weights do not diverge (notice that when $S_i(A) = S_i(B)$ the instrument is not relevant and $\bp{S_i(A) - S_i(B)}^{-2}$ is not finite).

\subsubsection{Trimming in Passive Control Experiments}

    Choose a threshold learning rate $\alpha^*$ and exclude observations with a learning rate $\alpha < \alpha^*$. Also, chose a threshold $s^*$ and exclude observations with $\bp{S_i(A) - X_i^0}^2 < s^*$  to ensure that the weights do not diverge (notice that when $S_i(A) = X_i^0$ the instrument is not relevant and $\bp{S_i(A) - X_i^0}^{-2}$ is not finite).

\subsection{Bandwidth Selection} \label{sec:estimation_bandwidth_choice}

Table \ref{tab:lls_bw_appendix} presents Local Least Squares (LLS) estimates across all six applications alongside the original paper estimates for comparison. For each application, I report LLS estimates using four different bandwidth choices to illustrate the bias-variance tradeoff inherent in nonparametric estimation methods.

In the all applications, the conditioning variable (the learning rate or belief update) is transformed to ranks and normalized to the unit inverval. Since the Epanechnikov kernel only has positive weight on the interval $\bp{-1,1}$, this makes the bandwidth directly interpretable as the share of observations that receive positive weight in each local regression. To be explicit, for a bandwidth $h$, use $K\left(\frac{R(\Delta X_i) - R(x)}{h/2}\right)$, where $R(\cdot)$ denotes the rank transformation and $K$ is the Epanechnikov kernel. For example, a bandwidth of 0.05 roughly means that 5\% of the data is used in each local regression; this is a parsimonious way to implement an adaptive bandwidth that gets larger in areas where there are fewer observations.

For the main analysis in the paper, the bandwidths range from $0.01 to 0.1$.
These bandwidths are small enough to minimize contamination from inappropriate comparisons across different treatment intensities, yet large enough to yield reasonably precise estimates. In most studies, the estimates are relatively stable across several bandwidths. More reassuringly, the CAPE curves are also qualitatively similar across bandwidths. For example, Figure \ref{fig:ss_cape_bw} shows that the CAPE estimates for \citet{setteleHowBeliefs22} have a consistent peak in the second quartile and estimates in Figure \ref{fig:KGC_emp_cape_bw} \citep{kumarEffectMacroeconomic23} consistently slope downwards.

Estimation in active and passive designs proceeds in multiple steps: first, estimate the learning rate $\alpha_i$ (or its rank); second, estimate the \q{local} regressions over the grid of learning rates; third, aggregate the local estimates by bins of the learning rate to estimate the CAPE (as in Figure \ref{fig:cape}) or over the entire grid to estimate the APE (as in Table \ref{tab:all_applications}). Estimation in the panel case also proceeds in multiple steps, but skips estimation of the learning rate and begins directly by estimating local regressions conditional on the change in beliefs. It is important that the bootstrap resampling takes place before the first step so that the resulting standard errors reflect the uncertainty associated with the entire procedure. All standard errors in this paper are estimated using 1000 iterations of the Bayesian bootstrap with $1\%$ of outliers dropped for stability \citet{hansenEconometrics22}.

\begin{table}[htbp]
\centering
\caption{LLS and Fixed Effects Estimates}
\label{tab:lls_bw_appendix}
{\small
\begin{tabular}{lcccc}
\toprule
& \multicolumn{4}{c}{\textsc{Panel A:} Panel Experiments} \\
& \multicolumn{4}{c}{\citet*{wiswallDeterminantsCollege15}} \\
\cmidrule(lr){2-5}
Coefficient    & 0.695   & 0.721  & 0.808   & 0.379 \\
Standard Error & (0.284) & (0.29) & (0.319) & (0.276) \\
Bandwidth      & 0.025   & 0.05   & 0.075   & 0.1 \\

\bottomrule
  \\
& \multicolumn{4}{c}{\citet*{armonaHomePrice19}} \\
\cmidrule(lr){2-5}
Coefficient & 1.716 & 1.8 & 1.64 & 1.69 \\
Standard Error & (0.377) & (0.387) & (0.384) & (0.367) \\
Bandwidth & 0.01 & 0.025 & 0.05 & 0.1 \\

\bottomrule
  \\
& \multicolumn{4}{c}{\textsc{Panel B:} Active Experiments} \\
& \multicolumn{4}{c}{\citet*{setteleHowBeliefs22}} \\
\cmidrule(lr){2-5}
Coefficient & 0.178 & 0.16 & 0.132 & 0.117 \\
Standard Error & (0.061) & (0.042) & (0.037) & (0.035) \\
Bandwidth & 0.005 & 0.01 & 0.025 & 0.05 \\

\bottomrule
  \\
& \multicolumn{4}{c}{\citet*{rothRiskExposure22}} \\
\cmidrule(lr){2-5}
Coefficient & 1.138 & 0.882 & 0.591 & 0.353 \\
Standard Error & (0.373) & (0.365) & (0.352) & (0.322) \\
Bandwidth & 0.05 & 0.075 & 0.1 & 0.15 \\

\bottomrule
  \\
& \multicolumn{4}{c}{\textsc{Panel C:} Passive Experiments} \\
& \multicolumn{4}{c}{\citet*{kumarEffectMacroeconomic23}} \\
\cmidrule(lr){2-5}
Coefficient & 1.368 & 1.787 & 2.036 & 2.214 \\
Standard Error & (0.457) & (0.465) & (0.538) & (0.589) \\
Bandwidth & 0.01 & 0.025 & 0.05 & 0.1 \\

\bottomrule
  \\
& \multicolumn{4}{c}{\citet*{cantoniProtestsStrategic19}} \\
\cmidrule(lr){2-5}
Coefficient & 0.182 & 0.18 & 0.18 & 0.179 \\
Standard Error & (0.236) & (0.164) & (0.133) & (0.12) \\
Bandwidth & 0.025 & 0.05 & 0.1 & 0.2 \\

\bottomrule
 \end{tabular}
} \vspace*{-1em}
\begin{quote} \small \textit{Notes:} This table presents estimates of the effect of beliefs on outcomes from all six replication studies. LLS estimates are presented for different bandwidth choices at four different bandwidth choices. In all applications, the conditioning variable is transformed to ranks; these bandwidths thus have intuitive interpretation as the share of the data used in each local regression. Standard errors are reported in parentheses. They are the standard deviation of the bootstrap distribution with $1000$ draws and $1\%$ of outliers dropped for stability \citep{hansenEconometrics22}.
\end{quote}
\end{table}

\clearpage

\begin{figure}[htb]
    \includegraphics[width=\linewidth]{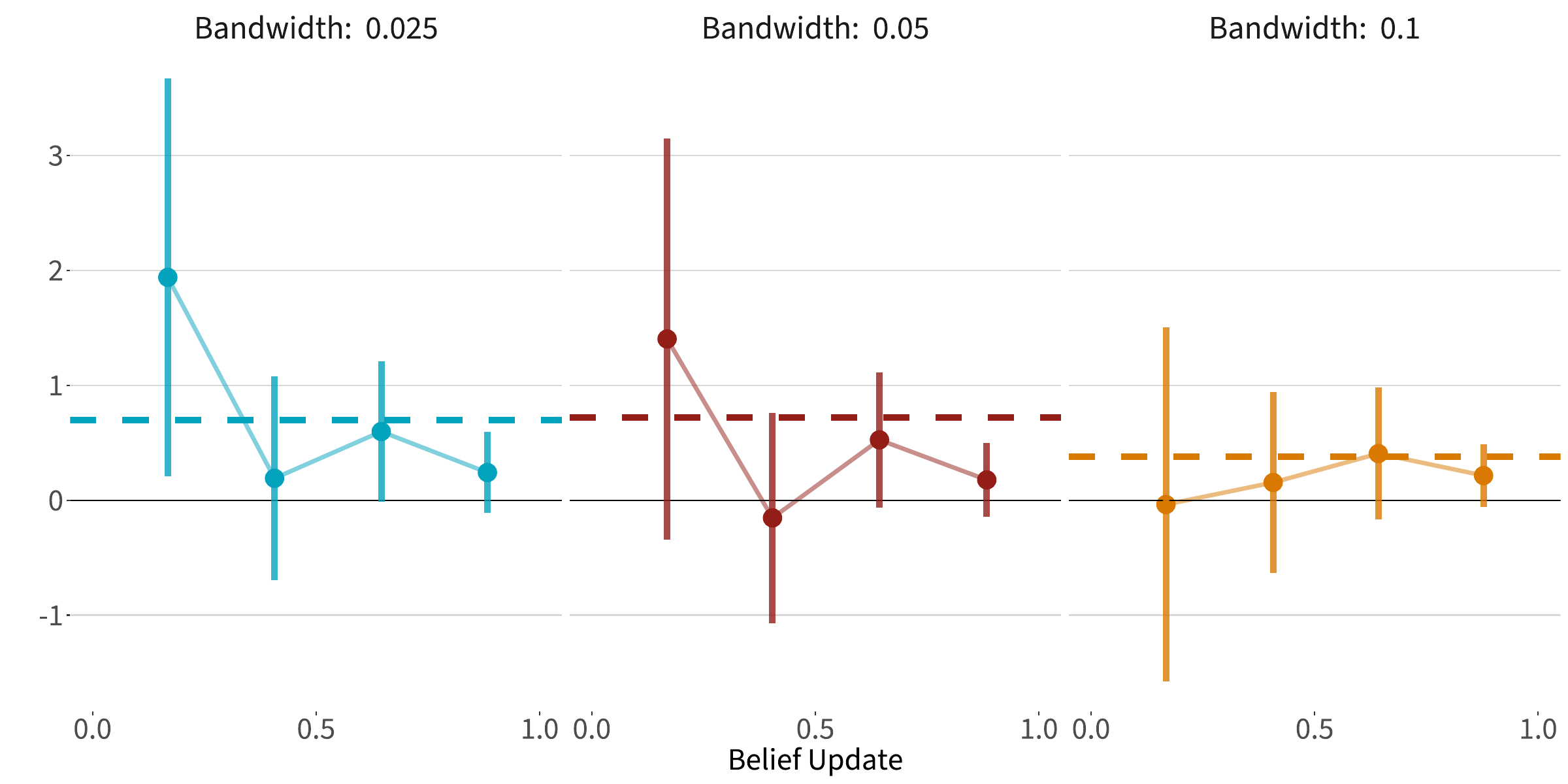}
    \caption{Conditional Average Partial Effects in \citet{wiswallDeterminantsCollege15}, Several Bandwidths}
    \label{fig:wz_cape_bw}
    \begin{quote}
        \textit{Notes:} This figure plots estimates of the conditional average partial effect $E[\tau_i|\Delta X_i=x]$ against the size of the belief update $x$. Each panel shows results for a different bandwidth choice. The dashed horizontal line in each panel shows the average partial effect (APE) estimated using that bandwidth.
       Confidence intervals displayed are twice the bootstrap standard errors. See Table \ref{tab:lls_bw_appendix} for the point estimate and standard error of the APE.
    \end{quote}
\end{figure}

\begin{figure}[htb]
    \includegraphics[width=\linewidth]{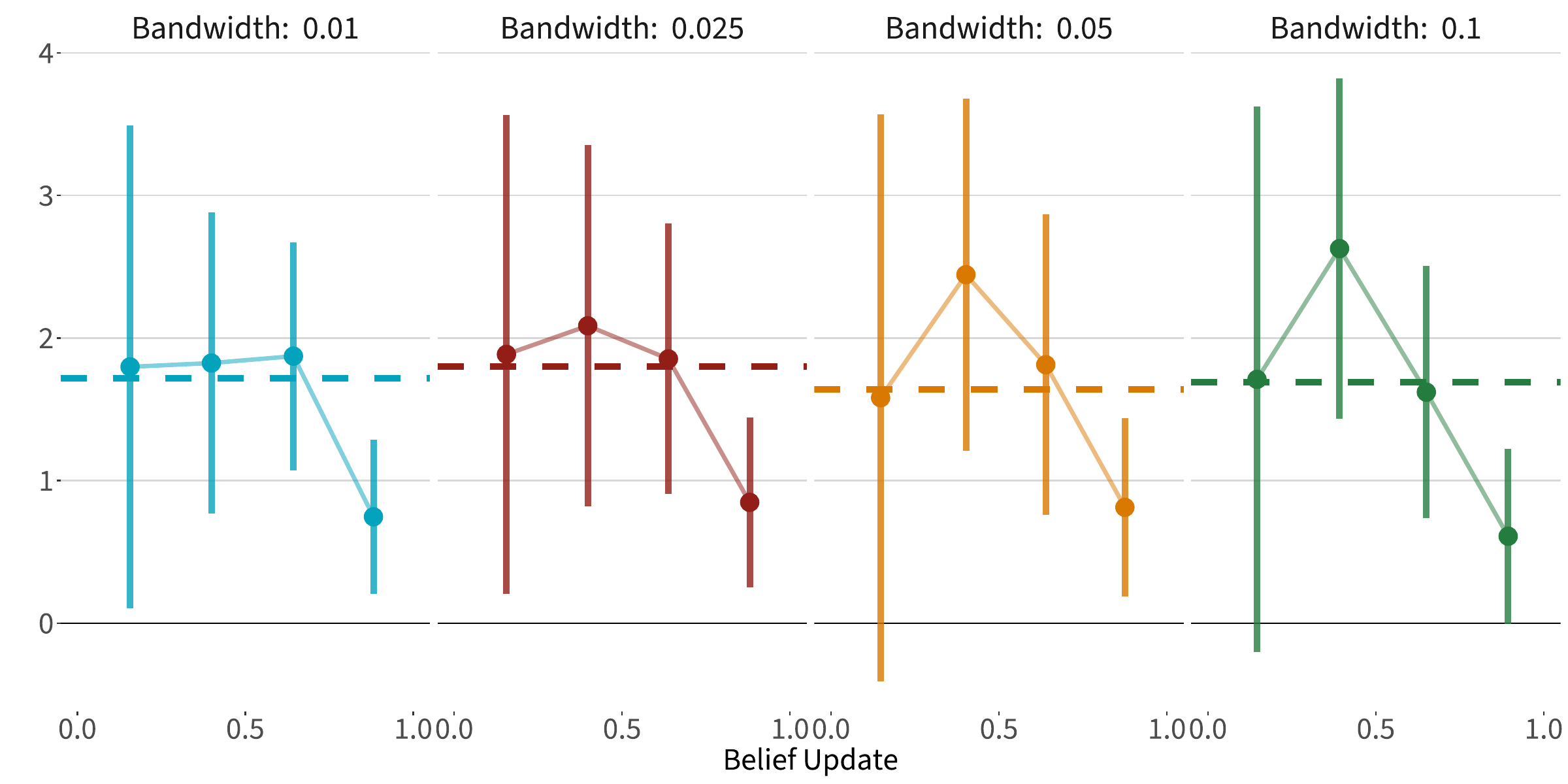}
    \caption{Conditional Average Partial Effects in \citet{armonaHomePrice19}, Several Bandwidths}
    \label{fig:afz_cape_bw}
    \begin{quote}
        \textit{Notes:} This figure plots estimates of the conditional average partial effect $E[\tau_i|\Delta X_i=x]$ against the size of the belief update $x$. Each panel shows results for a different bandwidth choice. The dashed horizontal line in each panel shows the average partial effect (APE) estimated using that bandwidth.
        Confidence intervals displayed are twice the bootstrap standard errors. See Table \ref{tab:lls_bw_appendix} for the point estimate and standard error of the APE.
    \end{quote}
\end{figure}

\begin{figure}[htb]
    \includegraphics[width=\linewidth]{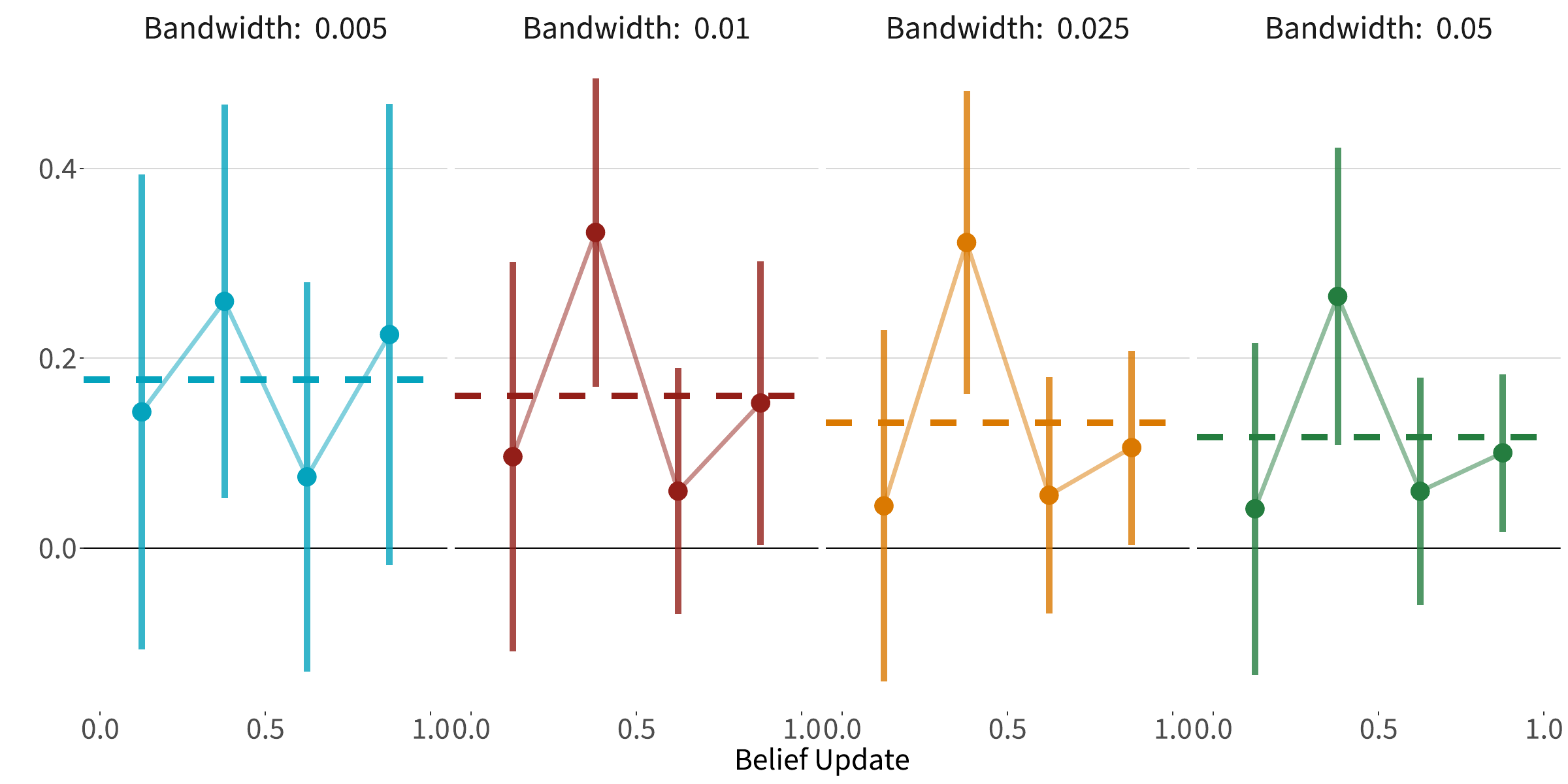}
    \caption{Conditional Average Partial Effects in \citet{setteleHowBeliefs22}, Several Bandwidths}
    \label{fig:ss_cape_bw}
    \begin{quote}
        \textit{Notes:} This figure plots estimates of the conditional average partial effect $\E \bs{\tau_i \mid \text{rank}\bp{\alpha_i}}$ the rank of the individual learning rate. Each panel shows results for a different bandwidth choice. The dashed horizontal line in each panel shows the average partial effect (APE) estimated using that bandwidth.
        Confidence intervals displayed are twice the bootstrap standard errors. See Table \ref{tab:lls_bw_appendix} for the point estimate and standard error of the APE.
    \end{quote}
\end{figure}

\begin{figure}[htb]
    \includegraphics[width=\linewidth]{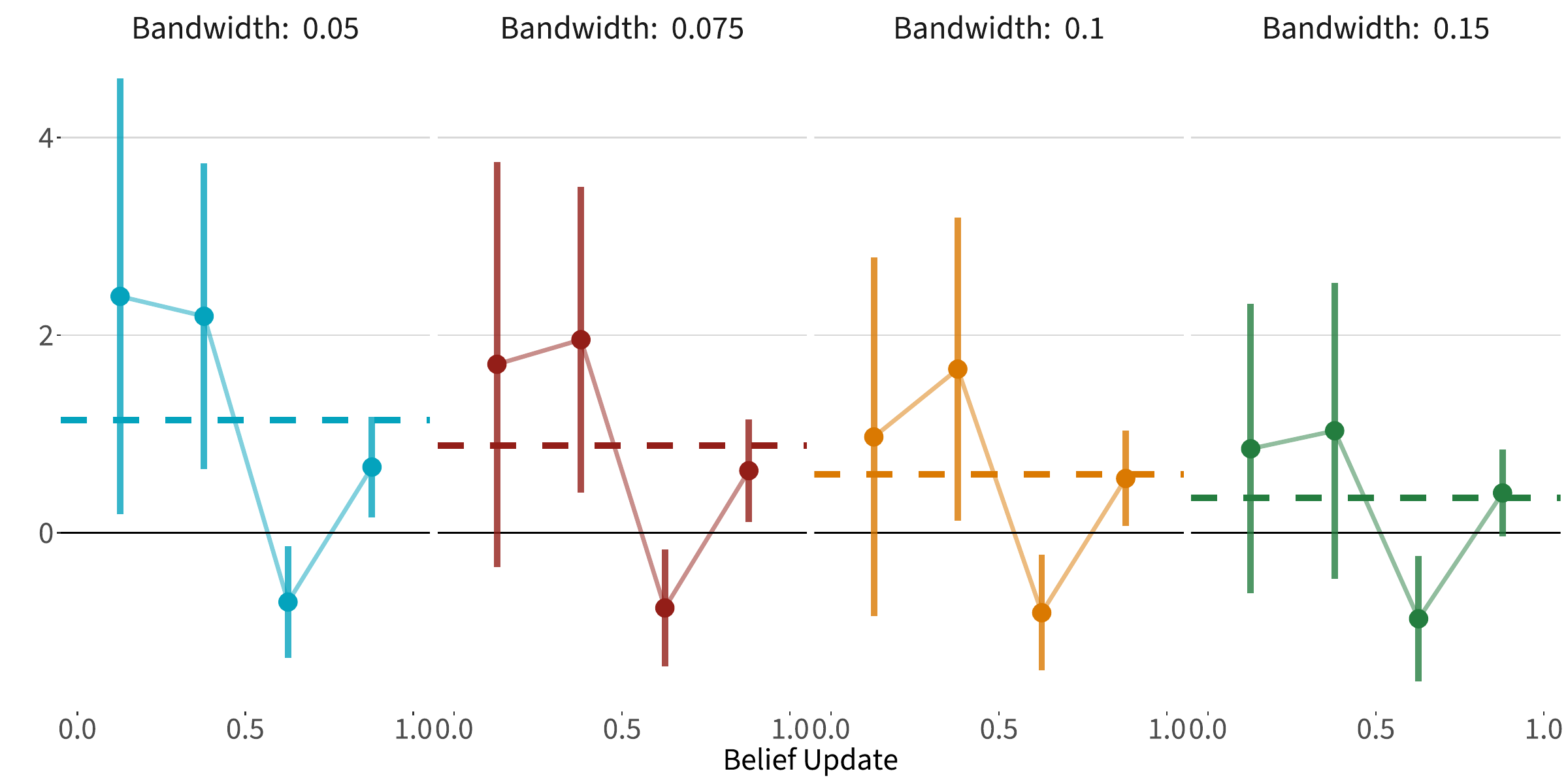}
    \caption{Conditional Average Partial Effects in \citet{rothRiskExposure22}, Several Bandwidths}
    \label{fig:rsw_cape_bw}
    \begin{quote}
        \textit{Notes:} This figure plots estimates of the conditional average partial effect $\E \bs{\tau_i \mid \text{rank}\bp{\alpha_i}}$ the rank of the individual learning rate. Each panel shows results for a different bandwidth choice. The dashed horizontal line in each panel shows the average partial effect (APE) estimated using that bandwidth.
        Confidence intervals displayed are twice the bootstrap standard errors. See Table \ref{tab:lls_bw_appendix} for the point estimate and standard error of the APE.
    \end{quote}
\end{figure}

\begin{figure}[htb]
    \includegraphics[width=\linewidth]{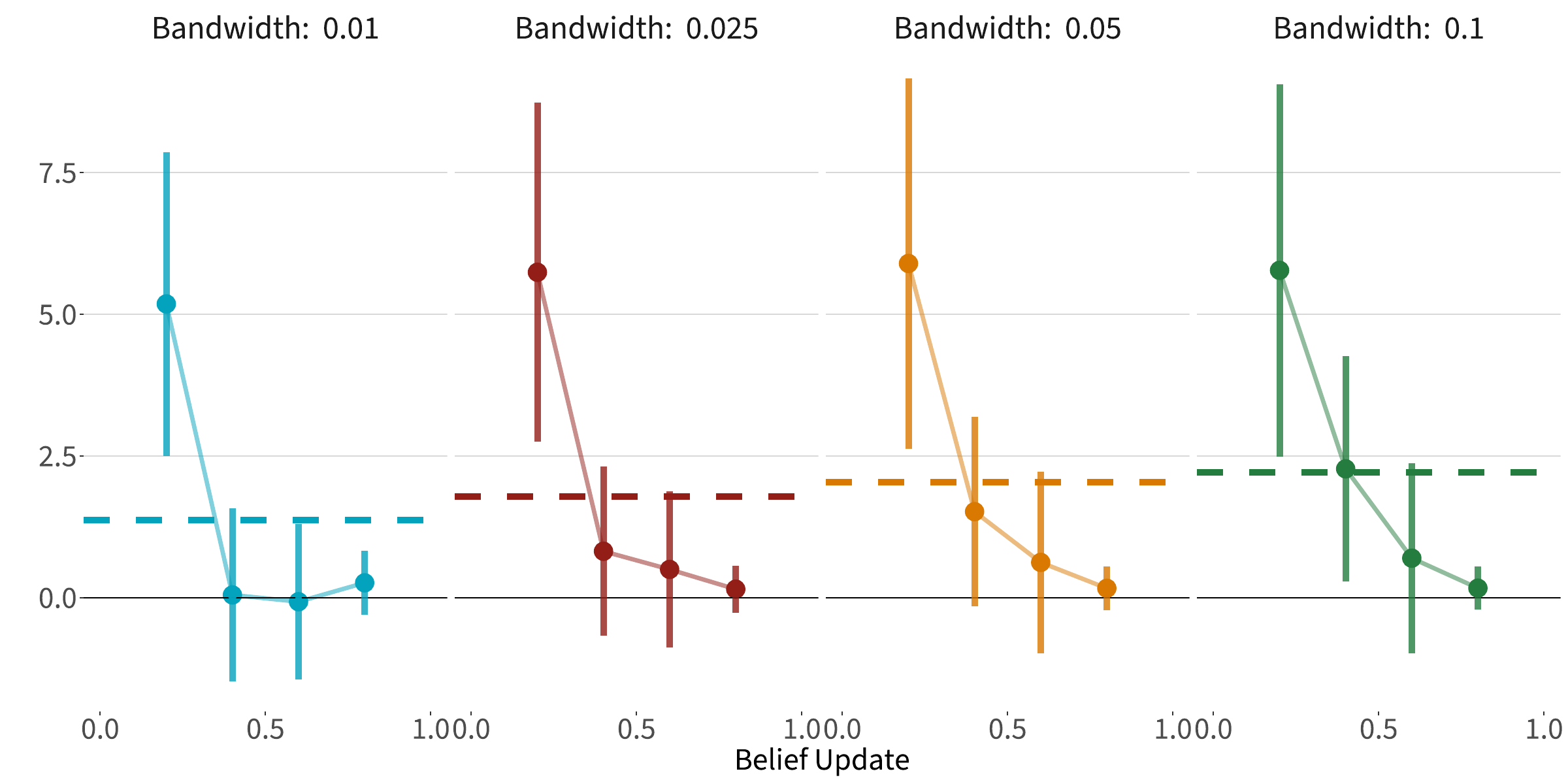}
    \caption{Conditional Average Partial Effects in \citet{kumarEffectMacroeconomic23}, Several Bandwidths}
    \label{fig:KGC_emp_cape_bw}
    \begin{quote}
        \textit{Notes:} This figure plots estimates of the conditional average partial effect $\E \bs{\tau_i \mid \text{rank}\bp{\alpha_i}}$ the rank of the individual learning rate. Each panel shows results for a different bandwidth choice. The dashed horizontal line in each panel shows the average partial effect (APE) estimated using that bandwidth.
        Confidence intervals displayed are twice the bootstrap standard errors. See Table \ref{tab:lls_bw_appendix} for the point estimate and standard error of the APE.
    \end{quote}
\end{figure}

\begin{figure}[htb]
    \includegraphics[width=\linewidth]{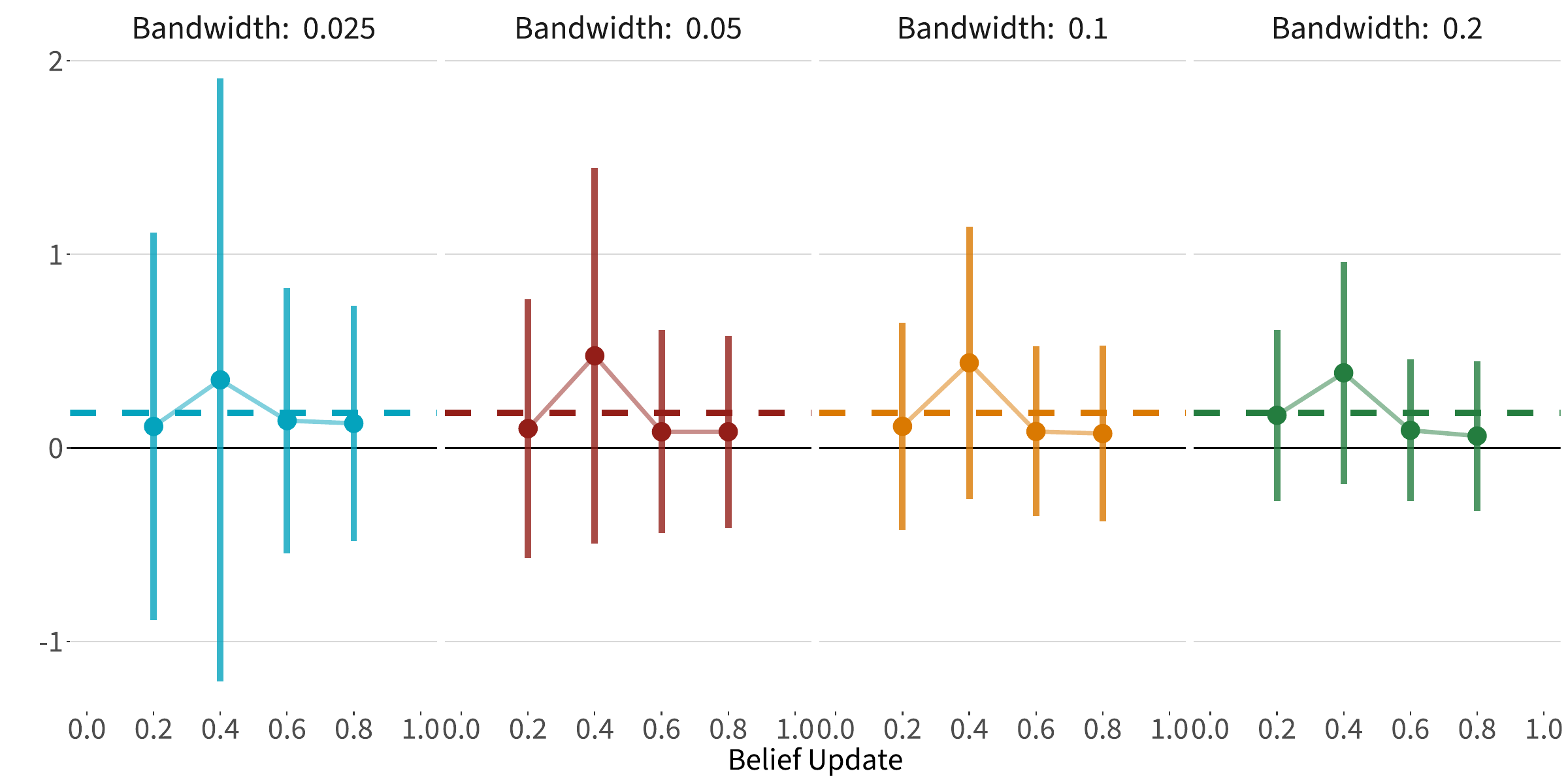}
    \caption{Conditional Average Partial Effects in \citet{cantoniProtestsStrategic19}, Several Bandwidths}
    \label{fig:CYYZ_cape_bw}
    \begin{quote}
        \textit{Notes:} This figure plots estimates of the conditional average partial effect $\E \bs{\tau_i \mid \text{rank}\bp{\alpha_i}}$ the rank of the individual learning rate. Each panel shows results for a different bandwidth choice. The dashed horizontal line in each panel shows the average partial effect (APE) estimated using that bandwidth.
        Confidence intervals displayed are twice the bootstrap standard errors. See Table \ref{tab:lls_bw_appendix} for the point estimate and standard error of the APE.
    \end{quote}
\end{figure}

  \clearpage

\section{Application Details} \label{app_application_details}

This section provides additional information about the key specifications under consideration in each of the six applications.

\subsection{Systematic Selection of Empirical Applications} \label{sec:paper_selection}

I identified papers for empirical reanalysis through a systematic search of top economics journals. On April 15, 2024, I searched the Web of Science database for papers published in the top five economics journals (American Economic Review, Econometrica, Journal of Political Economy, Quarterly Journal of Economics, Review of Economic Studies), the Review of Economics and Statistics, the four American Economic Journals, and AER: Insights. The search identified papers containing ``beliefs,'' ``information,'' or ``perception'' together with ``experiment'' or ``treatment'' in their title or abstract. This yielded 344 papers and 22 duplicates.\footnote{The Web of Science search query was
\texttt{
        (SO=(JOURNAL OF POLITICAL ECONOMY) OR SO=(AMERICAN ECONOMIC REVIEW) OR SO=(QUARTERLY JOURNAL OF ECONOMICS) OR SO=(REVIEW OF ECONOMIC STUDIES) OR SO=(ECONOMETRICA) OR SO=(REVIEW OF ECONOMICS "AND" STATISTICS) OR SO=(AMERICAN ECONOMIC JOURNAL* OR AMERICAN ECONOMIC REVIEW INSIGHTS)) AND (TI=(Belief OR Information OR perception) OR AB =(Belief OR Information OR perception)) AND (TI=(Experiment OR Treatment ) OR AB=(Experiment OR Treatment))}
}

I applied a hierarchical set of exclusion criteria to identify papers suitable for reanalysis. The initial screen excluded papers that were not information provision experiments (228 papers). This left 116 experiments, which I sorted by citation count and classified according to the first exclusion criterion each failed. I required papers to study how beliefs affect outcomes, not just how information affects beliefs (the first stage) or how information affects outcomes (reduced form). The experimental design had to follow a prior-treatment-posterior-outcome structure and be one of three types compatible with our estimator: panel experiments, active control experiments, or passive control experiments. For passive control designs, I additionally required that the study elicit the variance or uncertainty of participants' prior beliefs, which is necessary to model heterogeneity in belief updating. Finally, I required publicly available replication data.

I sought two examples of each experimental design type. After identifying six papers meeting all criteria (two panel, two active control, two passive control), I stopped screening. The remaining 61 papers had fewer citations than the least-cited included paper.

Table \ref{tab:paper_selection} shows the classification of all 344 papers. Among the 55 experiments I fully screened, the most common reasons for exclusion were incompatible experimental designs (25 papers), papers studying reduced-form effects of information without measuring beliefs (10 papers), and papers studying only belief updating without measuring outcome effects (8 papers).
These design-incompatible experiments are experiments related to beliefs or information but using designs other than the ones studied in the paper. For example, \citet{pallaisInefficientHiring14a} randomly assigns workers to recieve either detailed or coarse public evaluations and shows that inexperienced workers value this public information about their ability. One paper was excluded solely for lack of public replication data.

\begin{table}[h]
\centering
\caption{Systematic Paper Selection Results}
\label{tab:paper_selection}
\begin{tabular}{lrp{8.5cm}}
\toprule
Classification & Count & Description \\
\midrule
\multicolumn{3}{l}{\textit{Initial search results}} \\
Total papers identified & 344 & \\
\addlinespace[0.5em]
\multicolumn{3}{l}{\textit{Stage 1: Information provision experiment}} \\
Not experiment & 228 & Not an information provision experiment \\
Experiments to screen & 116 & \\
\addlinespace[0.5em]
\multicolumn{3}{l}{\textit{Iterative Step: Screen Until Two Examples Found In Each Design}} \\
Eliminated by citation cutoff & 61 & Fewer citations than least-cited included paper \\
Experiments fully screened & 55 & \\
\addlinespace[0.5em]
\multicolumn{3}{l}{\textit{Stage 2: Design compatibility and data availability}} \\
No belief measurement & 10 & Outcome effects only, no beliefs updating \\
No outcome measurement & 8 & Belief updating only, no outcome effects \\
Misc. design incompatible & 25 & Not prior-treatment-posterior-outcome  \\
Passive, no variance & 5 & Passive control without prior uncertainty \\
No replication data & 1 & Replication package not publicly available \\
\addlinespace[0.5em]
\textbf{Included in analysis} & \textbf{6} & \textbf{Met all criteria (2 panel, 2 active, 2 passive)} \\
\bottomrule
\end{tabular}
\begin{minipage}{\textwidth}
\vspace{0.5em}
\footnotesize
\textit{Notes:} This table shows the results of our systematic search of Web of Science conducted on April 15, 2024. The search covered papers published in the top five economics journals (American Economic Review, Econometrica, Journal of Political Economy, Quarterly Journal of Economics, Review of Economic Studies), Review of Economics and Statistics, all American Economic Journals, and AER: Insights. Search terms were ``beliefs,'' ``information,'' or ``perception'' combined with ``experiment'' or ``treatment'' appearing in title or abstract. Papers were sorted by citation count within each category. After identifying two examples of each experimental design type (panel, active control, passive control), I stopped screening; remaining experiments had fewer citations than the least-cited included paper.
\end{minipage}
\end{table}
 
\subsection{Application Details: \citet*{wiswallDeterminantsCollege15}}
\citet{wiswallDeterminantsCollege15} study how beliefs about future earnings affect how college students choose majors.  Their panel experimental design measures beliefs and outcomes before and after an information intervention.

\subsubsection{Setting}

In their experiment, undergraduate students were surveyed about their beliefs regarding future earnings, as well as population averages. They were also surveyed about their probability of graduating with a particular college major. After eliciting these prior beliefs, students received information about the true population distributions of these attributes. Finally, they reported revised beliefs about future earnings and college major choices.

\subsubsection{Specification of Interest}

The paper's main econometric specification is a first-difference regression of the change in stated probability of choosing a major on the change in beliefs about earnings. The authors normalize major choice and earnings relative to humanities/arts, thus the key first-differenced variables are

\begin{align}
\Delta Y_i &= \ln(\pi_{k,i,\text{post}}/\pi_{\bar{k},i,\text{post}}) - \ln(\pi_{k,i,\text{pre}}/\pi_{\bar{k},i,\text{pre}}) \\
\Delta X_i &= \ln(\omega_{k,i,\text{post}}/\omega_{\bar{k},i,\text{post}}) - \ln(\omega_{k,i,\text{pre}}/\omega_{\bar{k},i,\text{pre}})
\end{align}
where $\pi_{k,i}$ is the probability of majoring in field $k$ and $\omega_{k,i}$ is the expected earnings in field $k$ for individual $i$, with $\bar{k}$ representing humanities/arts. See page 814, equation 9 of \citet{wiswallDeterminantsCollege15} for details.

This specification follows column 3 of Table 6.B of \citet{wiswallDeterminantsCollege15}. This specification restricts to the sample of freshmen and sophomores (who are more able to adjust their major) and trims out outliers who update beliefs by more than $\$50,000$. This is the specification with the largest point estimate (and t-statistic) in Table 6.

\subsubsection{Implementing the LLS Estimator}

I also trim the sample to exclude very small updates (less than 0.05 in absolute value) that aren't exactly zero; this avoids regressions with very small variation in the regressors.\footnote{While point estimates are qualitatively similar without trimming away from zero, this trimming is important for the precision of estimates.} I also follow \citet{wiswallDeterminantsCollege15} and include fixed effects for college major in the local regressions.

 \subsection{Application Details: \citet*{armonaHomePrice19}}
\citet{armonaHomePrice19} study how past home price growth affects beliefs about home prices and how these expectations affect investment decisions. Their panel experimental design measures beliefs and outcomes before and after an information intervention.

\subsubsection{Setting}

In their experiment, participants in an online survey were first asked about their beliefs regarding past and future home price changes in their zip code. After eliciting these prior beliefs, the researchers provided a random subset of respondents with factual information about past local home price changes. They then re-elicited expectations about future price changes from all participants, creating an experimental panel. The outcome is constructed from a portfolio allocation task; participants were also asked to assign money to a savings account or a housing fund, both before and after the information treatment.

\subsubsection{Specification of Interest}

The paper's main econometric specification is a first-difference regression of the change in investment decisions (from the portfolio allocation task) on the change in beliefs about future home price growth.

Define $\Delta Y_i$ as the change in the percentage allocation to the housing asset and $\Delta X_i$ as the change in one-year-ahead home price expectations. For each individual $i$, we observe these changes directly as first differences:
\begin{align}
\Delta Y_i &= Y_{i1} - Y_{i0} \\
\Delta X_i &= X_{i1} - X_{i0}
\end{align}

This specification follows columns 5-7 of Table 10 of \citet{armonaHomePrice19}, with covariates omitted to focus on the key variable of interest.

\subsubsection{Implementing the LLS Estimator}

The sample selection criteria are as follows. As in column (7) of Table 10 of \citet{armonaHomePrice19}, the coefficient of interest is the coefficient on $\Delta X_i$ among the treated group; the control group is omitted from the regression. I also trim the sample to exclude very small updates (less than 0.025 in absolute value) that aren't exactly zero to avoid regressions with very small variation in the regressors. 
\subsection{Application Details: \citet{setteleHowBeliefs22}}
\citet{setteleHowBeliefs22} studies how beliefs about the gender wage gap affect support for policies aimed at reducing gender inequality. The active control experimental design provides all participants with information about the gender wage gap, but varies the information across treatment groups.

\subsubsection{Setting}

In the experiment, participants were first asked to report their beliefs about the gender wage gap. Then, participants were randomly assigned to see either a \q{high gap} truthful estimate (women earn 74\% of men's wages) or a \q{low gap} truthful estimate (women earn 94\% of men's wages). They were then asked to report their beliefs about the gender wage gap again after seeing the signal and were asked about their support for various gender-equality policies.

\subsubsection{Specification of Interest} \label{sec:apdx_ss_spec_details}

The paper's main econometric specification uses a two-stage least squares (TSLS) regression, where assignment to the \q{high gap} treatment serves as an instrument for posterior beliefs about the gender wage gap.
This specification follows column 7 of Table 5.C of \citet{setteleHowBeliefs22}. Posterior beliefs and the outcome are z-scored. The outcome in column 7 is a summary index constructed from demand for six gender-equality policies. The construction of the index is described in Online Appendix D.7 of \citet{setteleHowBeliefs22} as follows:

\begin{bquote}
To adjust for multiple inference, I follow \citet{andersonMultipleInference08} in applying a combined approach: First, I group the main outcome variables of interest into families and test for an overall treatment effect in a highly conservative way. Second, I test for a treatment effect on disaggregated outcomes within each family, allowing for more power in exchange for a small number of Type I errors. In the remainder of this section I describe the implementation of this combined approach and the intuition behind it (page 34, Online Appendix \cite{setteleHowBeliefs22}).
\end{bquote}

\subsubsection{Implementing the LLS Estimator} \label{sec:SS_implementing}

The point estimate in the original paper is negative and seeks to measure the effect of \q{women's relative earnings} on support for gender-equality policies. To make the discussion parsimonious across applications, we flip the sign of the belief variable so that point estimates are positive (unlike the original paper). The effect of interest can then be interpreted as the effect of \q{women's earnings gap} on support for gender-equality policies.

The sample selection criteria are as follows. We can only estimate the learning rate for individuals with $\text{prior} \neq \text{signal}$, so we exclude people with $\text{prior} = \text{signal}$. Additionally, the local regression is not identified for individuals with $\alpha = 0$, so we exclude them as well.\footnote{Directly dividing the belief update by the difference between the signal and the prior leads to very noisy estimates of the learning rate, which causes the LLS estimator to behave poorly in the bootstrap. Thus, for each individual in the sample, I take a kernel-weighted average of the belief update and the exposure to the signal and use that ratio to construct the learning rate. Intuitively, instead of constructing the learning rate from the raw prior and posterior, I construct it from smoothed versions of the prior and posterior.}
Finally, also exclude individuals with negative learning rates (those whose posterior is farther from the signal than their prior), as their updating doesn't follow reasonable updating patterns and thus the Bayesian learning structure does not hold on this sample.\footnote{\citet{vz_aeri} show that updating \q{towards the signal} is predicted by a much broader class of models than the Bayesian model. One reasonable interpretation is that these individuals are simply failing an \q{attention check}.}

As discussed in Appendix \ref{sec:linear_controls_lls}, it is sufficient to control non-parametrically for the learning rate $\alpha_i$ and to control linearly for the remaining elements of the control vector $\bs{S_i(A), S_i(B), X_i^0}$. Since the signals are common and $S_i(A) = 74 , S_i(B) = 94$ for all $i$, this simplifies further. The only remaining control variable is the prior $X_i^0$ and there is no need to reweight. Following \citet{setteleHowBeliefs22}, I include fixed effects for the elicitation subgroup, since this is the level of randomization. Other controls and sampling weights are omitted. The local regression is thus a regression of $Y_i$ on $X_i, X_i^0$ and elicitation subgroup fixed effects conditional on (the rank of) $\alpha_i$.

 \subsection{Application Details: \citet*{rothRiskExposure22}}
\citet{rothRiskExposure22} study how perceived exposure to macroeconomic risk affects households' demand for macroeconomic information. Their active control experimental design exploits sampling variation between two official census surveys to create exogenous variation in beliefs about exposure to unemployment risk.

\subsubsection{Setting}
In this experiment, participants first reported their prior beliefs about how the Great Recession affected unemployment rates among similar people. Then, participants were randomly assigned to receive truthful information about actual unemployment rate changes during the Great Recession based on data from either the American Community Survey (ACS) or the Current Population Survey (CPS). Sampling variation and procedural differences between these two surveys generate variation in the signals.

After receiving this information treatment, participants reported their posterior beliefs about their personal probability of becoming unemployed during the next recession. Finally, respondents chose between receiving expert forecasts about four different macroeconomic variables: recession likelihood, inflation, government bond returns, or government spending, or receiving no forecast at all.

\subsubsection{Specification of Interest}

The paper's main econometric specification uses a two-stage least squares (TSLS) regression where the difference in unemployment increase information between ACS and CPS data serves as an instrument for posterior beliefs about personal unemployment risk during the next recession. I replicate the main specification where the outcome variable is the probability of choosing to receive a recession forecast (multiplied by 100 so that the final estimates are in percentage point units). Since there is individual level variation in the potential signals, this estimand does not simplify to the expression given in \ref{eq:active_weights}. Instead, this estimand targets a weighted average of $\tau_i$ with weights $\omega_i \propto \alpha_i (S_i(A) - S_i(B))^2$.

More formally, the instrument is

\begin{equation}
T^{\Delta}_i \equiv \begin{cases}
S_i(A) - S_i(B) & \text{if } Z_i = A \\
S_i(B) - S_i(A)  & \text{if } Z_i = B
\end{cases}
\end{equation}

and the TSLS estimand is

\begin{equation}
\frac{\cov \bs{T^{\Delta}_i, Y_i}}{\cov \bs{T^{\Delta}_i, X_i}} = \E \bs{\tau_i \cdot \frac{\alpha_i (S_i(A) - S_i(B))^2}{\E{\alpha_i (S_i(A) - S_i(B))^2}}}
\end{equation}

\subsubsection{Implementing the LLS Estimator}

As in \citet{setteleHowBeliefs22}, we implement the LLS estimator using the two-step approach. The signals vary across participants based on their demographic characteristics, so we weight the local regressions by the inverse of the squared exposure $(S_i(A) - S_i(B))^{-2}$ to account for this variation in instrument strength.

The estimation of the learning rate and the sample restrictions are identical to \citet{setteleHowBeliefs22}, as discussed in \ref{sec:SS_implementing}. I use a smoothed estimate of the learning rate and exclude individuals with $\alpha \leq 0$. Additionally, since there are individual specific signals, I trim individuals with very small variation in the potential signals and require that $\bp{S_i(A) - S_i(B)}^2 > 0.25$. This ensures that the weights proportional to $(S_i(A) - S_i(B))^{-2}$ are well behaved.

The local regression is thus a regression of $Y_i$ on $X_i, X_i^0, S_i(A), S_i(B)$ conditional on (the rank of) $\alpha_i$, with weights proportional to $(S_i(A) - S_i(B))^{-2}$. The linear controls for $X_i^0, S_i(A), S_i(B)$, are sufficient to ensure that the residual variation is mean independent of the error term $U_i$. The weights ensure that each covariate group receives equal weight in the local regression so that the estimand retains its interpretation as an unweighted average.  
\subsection{Application Details: \citet*{kumarEffectMacroeconomic23}}
\citet{kumarEffectMacroeconomic23} study how firms' macroeconomic forecasts affect their economic decisions. The passive experiment provided a random subset of participants with a macroeconomic forecast.

\subsubsection{Setting}
In this experiment, participating firms were first asked to report their prior beliefs about GDP growth. Then, participants were then randomly assigned to one of three treatment groups receiving different types of information about macroeconomic forecasts, or to a control group receiving no information. Finally, they reported revised beliefs about GDP growth as well as actual firm decisions six months later.

Like \citet{vz_aeri}, I exclude the treatment groups that were designed to shift the second moment of beliefs and use only the first treatment group that provided information about the level of GDP growth.\footnote{As \citet{vz_aeri} also discuss, belief experiments with multiple information treatments that induce variation in both the level and the uncertainty of beliefs are delicate to interpret when effects of both the mean and the effect of the uncertainty are heterogeneous. In general, TSLS specifications with multiple endogenous variables can be difficult to interpret \citep{bhuller2SLSMultiple24}.} The analysis in this paper uses only comparisons between a single treatment arm and the control.

\subsubsection{Specification of Interest}

The main econometric specification I replicate is a simplified version of the system of equations given in equations 3 and 4'. Instead of using all treatment arms to instrument for both the posterior mean and posterior uncertainty, I use only the first treatment arm to instrument for the posterior mean. I interact the treatment indicator with the sign of the difference between the signal and the prior.\footnote{\citet{vz_aeri} also replicate these results and use only the first treatment arm. They show that results are similar in specifications that interact treatment with the actual difference between the signal and the prior and those that only interact it with the sign of the difference. Results can be different, however, in specifications that also include the un-interacted treatment indicator, since specifications can have negative weights.}
This specification is similar in spirit to the estimates in Table 3 of \citet{kumarEffectMacroeconomic23}.

\subsubsection{Implementing the LLS Estimator} \label{sec:kgc_implementation}

\citet{kumarEffectMacroeconomic23} elicit not only the mean of the prior belief, but also the variance. The implementation of the LLS estimator in this application thus follows Case 1 discussed in Section \ref{sec:ape_identification}. Under the assumption that individuals agree on the variance of the signal, the rank of the learning rate is simply the rank of the prior variance; conditioning on the rank of the prior variance is sufficient to condition on the learning rate.

I trim individuals with very small variation in the exposure to the signal and require that $\bp{S_i - X_i^0)}^2 > 0.01$. This ensures that the weights proportional to $\bp{S_i - X_i^0)}^{-2}$ are well behaved.

The local regression is thus a regression of $Y_i$ on $X_i, X_i^0$ conditional on (the rank of) ${\sigma^2_X}_i$, with weights proportional to $\bp{S_i - X_i^0)}^{-2}$. The linear controls for $X_i^0$, is sufficient to ensure that the residual variation is mean independent of the error term $U_i$. The weights ensure that the covariate groups recieve equal weight in the inner regression so that our estimand retains its interpretation as an unweighted average. To make the CAPE curves presented in Figure \ref{fig:cape} Panel C.i and Figure \ref{fig:KGC_emp_cape_bw} more comparable to those in other designs, I estimate $\E{ \text{rank}(\alpha) \mid \text{rank}({\sigma^2_X}_i)}$ on the treated group and use this for the x-axis of the binned estimates.
 
\subsection{Application Details: \citet*{cantoniProtestsStrategic19}}
\citet{cantoniProtestsStrategic19} study how  beliefs about others' participation in protests affect an individuals' own protest decisions. The passive experiment provided a random subset of participants with truthful information about the planned participation of their classmates.

\subsubsection{Setting}

In this experiment, participating students were asked to report prior beliefs about their classmates' participation in an upcoming political protest. Then, one day before the protest, a random subset of participants were provided with truthful information about the planned participation of their classmates. Finally, after the protest, they collected data on participants' actual protest behavior.

\subsubsection{Specification of Interest}

The paper's main econometric specification uses a two-stage least squares (TSLS) regression where treatment indicator, interacted with the sign of the difference between the prior and the signal, is an instrument for posterior beliefs. This specification targets a weighted average of $\tau_i$ with weights $\omega_i \propto \alpha_i \abs{S_i - X_i^0}$.

The TSLS estimand is

\begin{equation}
    \frac{\cov \bs{\text{sign}{\bp{S_i - X_i^0}} T_i,  Y_i }}{\cov \bs{\text{sign}{\bp{S_i - X_i^0}} T_i, X_i }}
\end{equation}

\subsubsection{Implementing the LLS Estimator}

\citet{cantoniProtestsStrategic19} collect a rich set of observables in their survey, which they use to predict prior beliefs in a supplemental analysis (Online Appendix Table A.5). The implementation of the LLS estimator in this application thus follows Case 2 discussed in Section \ref{sec:ape_identification}. Under the assumption that the counterfactual belief update in the passive control group can be predicted from rich observables, these estimates can be used to predict the (latent) learning rate in the control group. Then, the estimated learning rate can be used in the place of the observed learning rate in an active design.

I use the replication package provided by the authors to directly replicate the prediction exercise in Appendix Table A.5, directly predicting the learning rate instead of the prior belief. Then, I impose the same restrictions as in the active cases.
In particular, I restrict to learning rates strictly greater than zero. Like in \ref{sec:kgc_implementation}, I trim individuals with very small variation in the exposure to the signal and require that $\bp{S_i - X_i^0)}^2 > 0.01$.

The local regression is thus a regression of $Y_i$ on $X_i, X_i^0$ conditional on (the rank of) ${\wt{\alpha}}_i$, with weights proportional to $\bp{S_i - X_i^0)}^{-2}$. Recall that I use the notation $\wt{\alpha}_i$ to emphasize that the learning rate is predicted in the control group. The linear control for $X_i^0$, is sufficient to ensure that the residual variation is mean independent of the error term $U_i$. The weights ensure that the covariate groups recieve equal weight in the inner regression so that our estimand retains its interpretation as an unweighted average. To estimate standard errors, we use the empirical bootstrap with 1000 iterations.

\subsubsection{Discussion} \label{app_cyyz_discussion}

The TSLS estimate and the LLS estimate are both quite noisy, making it difficult to draw strong conclusions about the direction or magnitude of any difference.
However, if one takes the point estimates literally, it would suggest a different model of the dependence between belief updating and belief effects. Suppose that this is a setting where it is difficult for anyone to form precise beliefs so that uncertainty is widespread. Then, the relevant heterogeneity in updating may come from inattention: people who use the information in their decisions spend time carefully interpreting the signal and incorporating it into their beliefs. In constrast, people whose decisions don't depend on these beliefs may mostly ignore the signal and update their beliefs only slightly. A model where agents choose both how much information to acquire at baseline and how much to pay attention to new information as in \citet{fusterExpectationsEndogenous22} may be the appropriate theoretical generalization to unify the results across all six studies. An interesting task for future research would be to use the empirical tools provided in this paper to discipline models where the correlation between belief updating and the belief effects is ex ante ambiguous.
 
\clearpage \section{Endogenous Belief Formation Through Costly Information Acquisition} \label{sec:endog_info_model}

This section formalizes a model of endogenous information acquisition. When beliefs strongly affect decisions--think of a homeowner whose refinancing choices depend critically on house price expectations--individuals rationally invest in gathering precise information before any experiment takes place. These well-informed individuals update their beliefs only modestly when researchers provide new information, while those for whom the belief matters less start with noisier priors and update more dramatically. Since standard specifications weight individuals by the strength of their belief updating, they systematically under-weight precisely those people for whom beliefs matter most. I formalize this intuition by modeling how individuals trade off the cost of acquiring information against the risk of making decisions with imprecise beliefs. The resulting negative correlation between causal effects and belief updating leads to attenuated estimates in information provision experiments.

\subsection{General Model}
\label{sec:belief_formation}

People have a subjective belief distribution given by $F_i(\cdot)$. To make the analysis tractable, focus on belief distributions that can be characterized by their mean $\mu_i$ and variance $\sigma_i^2$, with $F_i$ belonging to a parametric family (e.g., normal distributions). People are uncertain about their beliefs, and this uncertainty about their beliefs generates uncertainty about the action that they would like to take. Let {$R(\tau_i, \sigma_i^2)$} denote the subjective risk or ex-ante regret (for example, the expected loss) that an individual with causal effect $\tau_i$  faces when their {belief variance is $\sigma_i^2$. {Note that $R$ depends on the distribution $F_i$ only through its variance $\sigma_i^2$, as the mean belief affects the level of the action but not the risk from uncertainty.}

We make the following assumptions on $R$.
First, uncertainty is costly: $\frac{\partial R}{\partial \sigma^2} \geq 0$, where $\frac{\partial R}{\partial \sigma^2} = 0 $ if and only if $\tau_i = 0$.
Second, since there is uncertainty in beliefs, it is costly to base behavior on these beliefs: $\frac{\partial R}{\partial \abs{\tau}} \geq 0$, where $\frac{\partial R}{\partial \abs{\tau}^2} = 0 $ if and only if $\sigma^2 = 0$.
Finally, uncertainty is more costly for people whose beliefs affect actions more: $\frac{\partial^2 R}{\partial \sigma^2 \partial \abs{\tau}} > 0$.

People make a decision to pay a cost $c > 0$ to obtain new information or to do nothing. There is an updating process such that the variance  of beliefs after viewing a signal $\sigma^2_+$ is less than the variance of the initial beliefs $\sigma^2$ . People then trade off the reduction in risk from the new information against the cost of the signal. Thus, when person $i$ has beliefs with variance $\sigma^2$, her loss can be given recursively by
\begin{equation}
    {V(\tau_i, \sigma^2)} = \min \left\{R(\tau_i, \sigma^2), V(\tau_i, \sigma^2_+) + c  \right\}
\end{equation}
Given the assumptions we have made on $R$, for any beliefs with $\sigma^2 > 0$, there is some threshold value $\tau^*$ such that people with $\abs{\tau_i} > \tau^*$ prefer to pay $c$ to update their beliefs. That such a threshold exists is guaranteed by the fact that $R(\tau_i, \sigma^2) = R(\tau_i, \sigma^2_+)$ when $\tau_i = 0$, which implies that $R(\tau_i, \sigma^2) < R(\tau_i, \sigma^2_+) + c$ at $\tau_i = 0$. However, since $\frac{\partial^2 R}{\partial \sigma^2 \partial \abs{\tau_i}} > 0$, we also know that $\frac{\partial R(\tau_i, \sigma^2)}{\partial \abs{\tau_i}} > \frac{\partial R(\tau_i, \sigma^2_+)}{\partial \abs{\tau_i}}$ since $\sigma^2 > \sigma^2_+$.

At $\tau_i = 0$, $R(\tau_i, \sigma^2)$} is below $R(\tau_i, \sigma^2_+) + c$. However,  $R(\tau_i, \sigma^2)$ is increasing faster than $R(\tau_i, \sigma^2_+)$ in $\abs{\tau_i}$ such that eventually these curves will cross. And since {$R(\tau_i, \sigma^2)$} is always increasing faster than {$R(\tau_i, \sigma^2_+)$} in $\abs{\tau_i}$, they will cross exactly once. Figure \ref{fig:belief-formation} illustrates this graphically. When beliefs are formed through such a process, people with larger causal effects of beliefs will have (weakly) more precise beliefs in equilibrium.

\subsection{A Simple Example with Quadratic Loss and Normal Beliefs}
\label{sec:toy_model_belief_formation}

This example illustrates how the general framework applies in an example where beliefs are normally distributed and the risk function takes a particularly tractable form.

Let $Y$ be the action (e.g., list price of a house) and $X$ denote beliefs (e.g., about the market value). People start with a prior belief distribution centered around $\pi_i$ with variance $\sigma_{X_0}^2$ so that their beliefs are represented by the normal $\N(\pi_i,\sigma^2_{X_0})$. For simplicity, $\sigma^2_{X_0}$ is common. Signals $S$ are drawn from a normal distribution $\N(\mu_S,\sigma^2_S)$. This is an assumption that people have the same information environment.

People are uncertain about their beliefs, and this uncertainty about their beliefs generates uncertainty about the action that they would like to take. People act to minimize the loss function $L_i(y,x) = D\bigl(y,Y_i(x)\bigr)$, for some distance function $D$, which is the disutility associated with taking action $y$ when $X = x$.
Intuitively, integrating $L_i(y,x)$ over the distribution of beliefs converts uncertainty about beliefs (i.e., what is the probability that $X=x$) into regret about actions (i.e., how far is the choice $y$ from $Y_i(x)$, which is optimal when $X=x$).
In this loss function, beliefs affect utility only through their effect on actions. There is no direct \q{psychic} cost of imprecise beliefs.

People choose $Y_i(x)$ following the rule $Y_i(x) = \tau_i x + U_i$, where $\tau_i$ and $U_i$ vary across individuals, and have quadratic loss $D(a,b) = \bp{a-b}^2$.
They act to minimize their expected loss, which is simply the expectation of $L_i(y,x)$ with respect to $X$ $\bp{\text{i.e. } \int L_i(y,x) dF(x)}$.

Let $\ol{X}$ denote the mean of the belief distribution. When beliefs are given by the normal $\N(\ol{X},\sigma^2_X)$, the choice of $Y$ that minimizes expected loss is simply $Y^* \equiv Y_i(\ol{X}) = \tau_i \ol{X} + U_i$. Use this to further simplify the expression for expected loss and write
\begin{align}
   \int L_i(Y^*,x) dF(x) &= \int D(Y_i(\overline{X}),Y_i(x)) dF(x) \\
   &= \int  \bp{\bp{\tau_i \ol{X} + U_i} - \bp{\tau_i x + U_i}}^2 dF(x)
    = \tau_i^2 \sigma^2_X
\end{align}
since $\E\bs{(\tau_i \ol{X} - \tau_i X)^2} = \tau_i^2 \mathrm{Var}(X) = \tau_i^2 \sigma_X^2$. Notice that with quadratic loss, the risk function takes the form $R(\tau_i, \sigma^2_X) = \tau_i^2 \sigma^2_X$, which satisfies the assumptions about $R$ given in Section \ref{sec:belief_formation}.

The disutility generated by uncertainty about $X$ is increasing in both the variance of the belief distribution and the magnitude of the causal effect of beliefs on the outcome. This expression allows us to study the information acquisition problem.

I endogenize belief formation by allowing people to pay a fixed cost $C$ to view a signal that is centered around the unknown true value. They then update beliefs following the normal-normal Bayesian learning formula. When a person's beliefs are given by $\N(\overline{X}, \sigma^2_X)$, her loss is given recursively by
\begin{align}
    V_i(\overline{X}, \sigma^2_X) &= \min \left\{\E_X[L_i(Y_i(\overline{X}),x)],
    \E_S\bs{ V_i(X'(s), \sigma^2_{X'})} + C  \right\}
\end{align}
where $\sigma^2_{X'} = \frac{\sigma^2_X\sigma_S^2}{\sigma^2_X + \sigma_S^2}$ is the posterior variance after observing signal $S$ and the expectation $\E[S]$ is with respect to the signal distribution. The benefit of the signal comes from the fact that the posterior variance is less than the prior variance as long as the prior distribution is not already degenerate. Notice that in this example, the value function depends on the belief distribution only through its variance $\sigma^2_X$, since the mean $\ol{X}$ affects the level of the optimal action but not the expected loss from uncertainty.

Solving this recursive problem gives the equilibrium condition
\begin{align}
    \tau_i^2 \sigma^2_X &=  \tau_i^2 \sigma^2_{X'} + C
\end{align}
In equilibrium, agents will be indifferent between paying the fixed cost to obtain new information and living with the uncertainty they have.\footnote{To ease exposition, I have ignored integer constraints that will, in general, prevent this from holding with equality. People will purchase signals until the next signal reduces their expected loss by less than the cost of the signal and will generally be strictly worse off if they buy another signal, not indifferent. This technicality makes exposition more cumbersome without any conceptual payoff.} Replacing $\sigma^2_{X'}$ with its definition, and recalling that $1-\frac{\sigma^2_S}{\sigma^2_S+\sigma^2_X} = \alpha_i$ yields the following equality
\begin{align}
    \alpha_i \tau_i^2 \sigma^2_X  =  C \label{eq:info_acquisition_LR}
\end{align}

Agents for whom the outcome is very sensitive to the beliefs ($\tau_i^2$ is very large) will update their information until $\sigma^2_X \alpha_i$ is small.\footnote{Notice that since $\alpha_i \equiv \frac{\sigma^2_X}{\sigma^2_S + \sigma^2_X}$, $\alpha_i$ and $\sigma^2_X$ move together. That is, holding fixed $\sigma^2_S$, an increase in $\sigma^2_X$ implies an increase in $\alpha_i$ and vice-versa.} On the other hand, agents for whom the outcome is not sensitive to beliefs ($\tau_i^2$ is small) will stop after seeing fewer signals, so that $\sigma^2_X \alpha_i$ is relatively large.

This simple model illustrates how the causal relationship of interest affects the formation of beliefs before the experiment takes place. People whose actions depend more on their beliefs will be more willing to pay to obtain new information, and will therefore have more precise beliefs. In a Bayesian updating model, people with more precise beliefs will be less responsive to new information. In this way, the amount of variation in beliefs that can be induced by experimentally providing new information is directly depends on the causal effects of interest.

\subsection{Using Models of Belief Formation and Updating to Interpret TSLS Estimates} \label{sec:model_tsls_discussion}

The class of parameters that are targeted by existing standard specifications depend not only on the causal effects of beliefs on outcomes $\tau_i$, but also on heterogeneity in the way that beliefs are updated in response to new information.

In the model proposed in this section, beliefs are formed endogenously through a process of costly information acquisition. In Appendix \ref{sec:toy_model_belief_formation}, I solve a special case of this model where the subjective risk is given by the expected quadratic loss $R(\tau_i, \sigma^2) = \tau_i^2 \sigma^2$. Parameterizing the loss function makes it possible to solve analytically for the learning rate $\alpha_i$ and variance of the prior $\sigma^2_i$ as a function of the causal effects of beliefs $\tau_i$.

People have inaccurate and imprecise beliefs precisely because they have small individual partial effects (small $\abs{\tau_i}$); when beliefs are an important determinant of the behaviors  (large $\abs{\tau_i}$), people exert effort to form accurate and precise beliefs. In this environment, parameters with weights proportional to the strength of the shift in beliefs will be attenuated and underestimate the magnitude of the average effect.

Alternative models of the relationship between belief updating and the effects of beliefs on behaviors can be used to relate causal parameters estimated using standard specifications to the APE. For example, \citet{fusterExpectationsEndogenous22} allow variation in the learning rate to come from a more complicated model that adds dynamics of rational inattention to costly information acquisition.

\clearpage

\begin{figure}[!tb]
    \centering
    \caption{People with Large Effects of Beliefs $\tau_i$ Form Precise Beliefs}

   \begin{tikzpicture}[domain=0:3,x =.3*\textwidth, y=.1*\textwidth]
\fill[seq0!20] (0,0) -- (1,0) -- (1,3) -- (0,3) -- cycle;
  \fill[seq1!20] (1,0) -- (2,0) -- (2,3) -- (1,3) -- cycle;
  \fill[seq2!20] (2,0) -- (3,0) -- (3,3) -- (2,3) -- cycle;

    \draw[->, line width=1pt] (-0.1,0) -- (3.1,0) node[right] {$\abs{\tau_i}$};
    \draw[->, line width=1pt] (0,-.1) -- (0,3.1) node[right] {Disutility (Risk)};

  \draw[line width=1.5pt, color=seq0, domain=0:1.8, variable=\x] plot ({\x},{\x^2}) node[left] {{$R(\tau_i, \sigma^2)$}};
  \draw[line width=1.5pt, color=seq1, domain=0:2.34, variable=\x] plot ({\x},{.5+.5*\x^2}) node[left] {$R(\tau_i, \sigma^2_+) + c$ };
  \draw[line width=1.5pt, color=seq2, domain=0:2.45, variable=\x] plot ({\x},{1+.375*\x^2}) node[right] {$R(\tau_i, \sigma^2_{++}) + 2c$};

  \draw[color = black!50, dashed, line width=1.5pt] (1,0) -- (1,3);
  \draw[color = black!50, dashed, line width=1.5pt] (2,0) -- (2,3);
    \node[below] at (1,0) {$\tau^*$};
    \node[below] at (2,0) {$\tau^{**}$};

    \node[align=left, below right]  at (0,3) {No signal \\ $\abs{\tau_i} < \tau^* $ \so $ \sigma^2_i = \sigma^2  $};
    \node[align=left, below right] at (1,1) {One signal \\ $\tau^* < \abs{\tau_i} < \tau^{**} $ \so $ \sigma^2_i = \sigma^2_+$};
     \node[align=left, below right] at (2,2) {Two signals \\ $\tau^{**} < \abs{\tau_i}  $ \so $ \sigma^2_i = \sigma^2_{++}$};

\end{tikzpicture}

    \label{fig:belief-formation}
        \begin{quote}
        \textit{Notes}:
    This figure plots the loss as a function of $\abs{\tau_i}$ after seeing no signals, one signal, and two signals. The assumptions on $R_i$ ensure that each pair of lines crosses exactly once. Since $R(\tau_i, \sigma^2) = R(\tau_i, \sigma^2_+)$ when $\tau_i = 0$, $R(\tau_i, \sigma^2) < R(\tau_i, \sigma^2_+) + c$. If $\sigma^2_{++} > 0$, these curves are all strictly increasing in $\abs{\tau_i}$ by assumption. Additionally, since $\sigma^2 > \sigma^2_{+} > \sigma^2_{++}$, then $R(\tau_i, \sigma^2)$ is steeper than $R(\tau_i, \sigma^2_+)$, which is steeper than $R(\tau_i, \sigma^2_{++})$ by the assumption that \footnotesize $\frac{\partial^2 R}{\partial \sigma^2 \partial \abs{\tau_i}}$  \normalsize $ > 0$.
    \end{quote}
\end{figure}

\clearpage
 
\clearpage \section{Information Experiments and the TSLS Estimator}
\label{app_tsls}

This appendix provides discussion of the interpretation of TSLS estimators in information provision experiments. The challenges with obtaining unconditional monotonicity motivate the representative specifications discussed in Section \ref{sec:lit_weights}, which have non-negative weights under a weaker conditional monotonicity assumption. While the weighted average interpretation of TSLS estimands is well-established \citep{angristTwoStageLeast95}, this section examines the specific implications for information experiments and relates them to a workhorse learning rate updating assumptions. Section \ref{sec:no_priors} provides a novel strategy to ensure non-negative weights when priors are not elicited.

\subsection{The Reduced Form Effect of Information Provision} \label{app_info_reduced_form}

In active designs, the reduced form effect of treatment is the effect of being assigned to see the signal in arm $A$ rather than the signal in arm $B$. In passive designs, this is the effect of being assigned to see new information. Consider the simple OLS regression of the outcome $Y_i$ on the treatment indicator $T_i \equiv \1\bc{Z_i = A}.$
\begin{align}
    \beta^{RF} &\equiv \frac{\cov \bs{T_i,Y_i}}{\var \bs{T_i}} \label{eq:RF_regression} \\
    &= \E \bs{\tau_i\bp{X_i(A) - X_i(B)}} \label{eq:RF_weights}
\end{align}
The reduced form effect of assignment to arm $A$ on the outcome is the expectation of the individual effect of beliefs on behaviors $\tau_i$ scaled by the individual effect of the information treatment on beliefs $X_i(A) - X_i(B)$. If all $\tau_i$ have the same sign, the reduced form effect of treatment assignment on the outcome will be informative of the sign of the effect of beliefs on behaviors only if the $X_i(A) - X_i(B)$ are all positive or all negative. If the first stage effect on beliefs is positive for some people and negative for others, then the average effect of the information treatment on behaviors can be close to zero, even if the effect of beliefs on behaviors is large and the individual first stage effects of the information treatment on beliefs are large.

\subsubsection{From the Effect of Information to the Effect of Beliefs}

As \citet{giaccobassoWhereMy22} note, reduced form estimates can be difficult to interpret since they combine the causal effects of beliefs on behaviors with the first stage effects of the information provision on beliefs. The reduced form can therefore be small if beliefs have only a weak effect on behavior, or if the information provision has only a weak effect on beliefs.

The reduced form is most directly policy-relevant when the counterfactual of interest concerns information provision per se rather than belief changes more generally. However, when the relationship of interest is the effect of beliefs on behavior, researchers typically normalize the reduced form effect by the first stage effect and report TSLS estimates.

\subsubsection{Constructing TSLS Estimates} \label{app_wald_tsls}

To motivate the specifications in Section \ref{sec:lit_weights}, we consider the simplest TSLS estimand as that directly uses treatment assignment $T_i$ to instrument for beliefs.
\begin{align}
    \beta^{TSLS} &\equiv \frac{\beta^{RF}}{\beta^{FS}} = \frac{\cov \bs{T_i,Y_i}}{\cov\bs{T_i,X_i}}   \label{eq:Ti_estimand} \\
    \intertext{where $\beta^{FS} \equiv \cov\bs{T_i,X_i}/\var\bs{T_i}$. For the binary treatment indicator, this becomes}
     \beta^{TSLS} &= \frac{\E \bs{Y_i \mid T_i= 1} - \E \bs{Y_i \mid T_i= 0}}{\E \bs{X_i \mid T_i= 1} - \E \bs{X_i \mid T_i= 0}} \label{eq:Ti_wald} \\
    \intertext{Substituting the linear outcome equation \eqref{eq:outcome} yields}
    \beta^{TSLS} &= \frac{\E \bs{\tau_i\bp{X_i(A) - X_i(B)}}}{\E{\bs{X_i(A) - X_i(B)}}} \label{eq:tsls_wald}
\end{align}

In the presence of heterogeneous effects, TSLS does not generally recover the average of the individual treatment effects. The TSLS coefficient depends on the covariance between individual belief effects $\tau_i$ and the first stage variation $X_i(A) - X_i(B)$:
\begin{align}
    \frac{\E \bs{\tau_i \bp{X_i(A) - X_i(B)}}}{\E \bs{\bp{X_i(A) - X_i(B)}}} &=  \E \bs{\tau_i}  + \frac{\cov \bs{\tau_i, \bp{X_i(A) - X_i(B)}}}{\E \bs{\bp{X_i(A) - X_i(B)}}} \label{eq:tsls_wald_covariance}
\end{align}

The covariance term is the \q{bias} relative to the APE $\E \bs{\tau_i}$ and motivates the LLS estimator developed in Section \ref{sec:lls}.

\subsection{Unconditional Instrument Monotonicity and Bayesian Updating} \label{app_tsls_bayes}

The weights derived in Section \ref{app_wald_tsls} are non-negative when unconditional monotonicity holds. This section examines when Bayesian updating ensures monotonicity across different experimental designs.

\subsubsection{Monotonicity in Active Designs} \label{app_active_monotonicity}

In active designs, monotonicity follows directly from Bayesian updating when signals are ordered such that $S_i(A) \geq S_i(B)$. Since $X_i(A) - X_i(B) = \alpha_i(S_i(A) - S_i(B))$ and $\alpha_i \in (0,1)$ under Bayesian updating, the sign of the first stage is determined by $\text{sign}(S_i(A) - S_i(B))$. The immediacy of monotonicity in active designs should be considered one advantage of this design relative to passive designs.

\subsubsection{Monotonicity in Passive Designs}

In passive designs, unconditional monotonicity requires that $S_i(A) - X^0_i$ has the same sign for all participants—either the signal is above everyone's prior or below everyone's prior. This is often empirically implausible; in all six empirical examples considered in this paper, we observe participants with priors both above and below the signal. This is why the simple specification \eqref{eq:Ti_estimand} is not widely used in practice; instead researchers use one of two main strategies to ensure positive weights.

\subsection{Strategies for Ensuring Non-Negative Weights in Passive Designs}

When unconditional monotonicity fails, researchers can construct specifications with non-negative weights by incorporating information about priors and signals.

\subsubsection{Sample Splitting Approach} \label{app_tsls_split}

Researchers can split the sample based on whether the signal is above or below each participant's prior, then estimate separate TSLS regressions within each subsample. For participants with $S_i(A) - X^0_i > 0$:
\begin{align}
    \beta^{\text{split}}_+ &= \frac{\cov \bs{T_i, Y_i \mid S_i(A) - X^0_i > 0}}{\cov \bs{T_i, X_i \mid S_i(A) - X^0_i > 0}} \\
    &= \E\left[\tau_i \cdot \frac{\alpha_i |S_i(A) - X^0_i|}{\E[\alpha_i |S_i(A) - X^0_i| \mid S_i(A) - X^0_i > 0]} \mid S_i(A) - X^0_i > 0\right]
\end{align}

A symmetric expression applies for $S_i(A) - X^0_i < 0$. Both specifications yield non-negative weights under Bayesian updating since $\alpha_i > 0$.

\subsubsection{Exposure-Weighted Instruments} \label{app_tsls_exposure}

An example of the exposure-weighted instrument is presented in Section \ref{sec:main_passive_exp_weights}.

\[ \wt{T}^{ex}_i \equiv (T_i - \E[T_i])(S_i(A) - S_i(B)) \]

The recentering is implicit since in practice researchers use the interaction as an instrument and control for the uninteracted exposure. Recall the notational device that in the passive design $S_i(B) = X_i^0$. These weights proportional to $\alpha_i (S_i(A)-X^0_i)^2$ are non-negative under Bayesian updating and in a general class of updating models when the monotonicity assumption holds: $\text{sign}(X_i(A)-X_i(B)) = \text{sign}(S_i(A)-S_i(B))$.

\citet{vz_aeri} show that implementation of these specifications requires care, as including both the exposure-weighted instrument and the treatment indicator can result in misspecification.

\subsection{Implementation When Priors Are Unobserved} \label{sec:no_priors}

Some experiments do not elicit prior beliefs directly. Under Bayesian updating, the direction of the belief update can be inferred from the posterior belief and the signal. If the posterior lies between the prior and signal, then $\text{sign}(S_i(A) - X_i) = \text{sign}(S_i(A) - X^0_i)$, allowing sample splitting even when  priors are unobserved.
This assumption identifies the same causal parameters that are targeted by $\beta^{\text{split}}_+$ and $\beta^{\text{split}}_-$ in Appendix \ref{app_tsls_split}. 

Since the control group that is not shown a signal, we directly observe their prior: recall that $X_i(B) = X^0_i$ in passive designs. Since the signal is known, we can directly condition on the sign of $(S_i(A) - X^0_i)$. The prior for the treated group is unknown and we observe only $X_i(A)$. But since we can rewrite the potential outcome equation in \ref{eq:beliefs_PO} as

$$S_i(A) - X_i(A)  = (1 - \alpha_i)(S_i(A) - X^0_i)$$

and since $\alpha \in \bp{0,1}$ then

$$ S_i(A) - X_i(A) > 0 \iff (S_i(A) - X^0_i)  > 0$$

We used the Bayesian updating structure, but note this could be relaxed to include any model of updating such that the posterior lies between the prior and the signal.

Thus, although the regressions in Section \ref{app_tsls_split} are not feasible since they use the prior to split the sample, the following regressions are feasible and equivalent.

\begin{align}
    \beta^{\text{split}}_+ = \wt{\beta}^{\text{split}}_+  &\equiv   \frac{\cov \bs{T_i,Y_i \mid S_i(A)-X_i > 0}}{\cov \bs{T_i,X_i \mid S_i(A)-X_i > 0}} \label{eq:split_wald_noprior_reg_up} \\
    \beta^{\text{split}}_- = \wt{\beta}^{\text{split}}_-  &\equiv   \frac{\cov \bs{T_i,Y_i \mid S_i(A)-X_i < 0}}{\cov \bs{T_i,X_i \mid S_i(A)-X_i < 0}} \label{eq:split_wald_noprior_reg_down}
\end{align}

\end{document}